\documentclass[submission,copyright,creativecommons]{eptcs}
\providecommand{\event}{GandALF 2021} % Name of the event you are submitting to
\usepackage{breakurl}             % Not needed if you use pdflatex only.
\usepackage{underscore}           % Only needed if you use pdflatex.

\usepackage{listings}
\lstdefinelanguage{pseudo}{morekeywords={init,with,or,if,then,else,fi,and,not,while,do,od,distinct,
    case, goto,local,algorithm, function, for, each, times, from, to,
    variables, procedure, recursive, return},
  morecomment=[l]{//}, morecomment=[s]{/*}{*/},
  mathescape=true,escapechar={@},
  basicstyle=\sffamily\small,
  commentstyle=\itshape\rmfamily\small,
  keywordstyle=\sffamily\bfseries\small
}
\usepackage{mathtools}
\usepackage{amsfonts,amssymb,amsmath, amsthm}

\usepackage{etoolbox}  % http://ctan.org/pkg/etoolbox
\usepackage{enumitem} % decorating enumerations

\usepackage{stmaryrd}
\usepackage{bm}
\usepackage{hyperref}
\usepackage{environ}
\usepackage{csvsimple}
\usepackage{longtable} % table-long.tex
\usepackage[normalem]{ulem} % underline

\usepackage{fourier-orns} % use \bomb
\usepackage{wasysym} % use \clock
\usepackage{marvosym} % use \Telefon
\usepackage{bbding} % use! \Checkmark, \XSolidBrush

\usepackage{etoolbox}  % http://ctan.org/pkg/etoolbox
\usepackage{enumitem} % decorating enumerations

\usepackage{enumerate}
\usepackage{mathrsfs}
\usepackage{mathtools}
\usepackage{appendix}
\usepackage{comment}
\usepackage{tikz}
\usepackage{array}
\usetikzlibrary{calc}
\usepackage{multicol}

\usetikzlibrary{arrows,shapes,snakes,automata,backgrounds,petri,positioning}
\usetikzlibrary{fit,backgrounds} 
\definecolor{processblue}{cmyk}{0.96,0,0,0}

\newcommand{\be}{\begin{enumerate}}
\newcommand{\ee}{\end{enumerate}}
\newcommand{\bc}{\begin{center}}
\newcommand{\ec}{\end{center}}

\newcommand{\bi}{\begin{itemize}}
\newcommand{\ei}{\end{itemize}}

\newcommand{\act}{\xrightarrow}

\newtheorem{theorem}{Theorem}

\newtheorem{prop}[theorem]{Proposition}

\newtheorem{definition}[theorem]{Definition}
\newtheorem{example}{Example}

\usepackage{hyperref}

\newcommand\slice[2]{#1{\raise-.5ex\hbox{\ensuremath|}}_{#2}}

\newcommand{\defeq}{\stackrel{\scriptscriptstyle\text{def}}{=}}

\newcommand{\N}{\mathbb{N}}                    %% Natural numbers
\renewcommand{\vec}[1]{\bm{#1}}                %% Vectors
\newcommand{\set}[1]{\left\{#1\right\}}        %% Custom set notation
\newcommand{\multiset}[1]{\Lbag#1\Rbag}        %% Bag notation
\renewcommand{\vec}[1]{\bm{#1}}                %% Bold vectors
\newcommand{\norm}[1]{\lVert#1\rVert}          %% Norm
\newcommand{\support}[1]{\norm{#1}} %%complement

\usepackage{thm-restate}

\newcommand{\net}{\mathcal{N}}

\newcommand{\trans}[1]{\xrightarrow{#1}}       %% Transition relation
\newcommand{\pre}{\mathit{pre}} %% 
\newcommand{\post}{\mathit{post}} %% 
\newcommand{\prestar}{\mathit{pre}^*}
\newcommand{\poststar}{\mathit{post}^*}

\renewcommand{\norm}[1]{\| {#1} \|}

\newcommand{\unorm}[1]{\|{#1}\|_u}
\newcommand{\lnorm}[1]{\|{#1}\|_l}
\newcommand{\sem}[1]{\llbracket{#1}\rrbracket} %% Semantic of CC

\newcommand{\cube}{\mathcal{C}}
\newcommand{\cC}{\Gamma}      %% Counting constraint
\newcommand{\cSet}{\mathcal{S}}  %% Counting set
\newcommand{\RBN}{\mathcal{R}}
\newcommand{\ASMS}{\mathcal{P}}
\newcommand{\oper}{\mathtt{op}}
\newcommand{\data}{\textit{data}}

\makeatletter

\newcommand{\eqxrightarrow}[2]{%
  \mathop{%
    \vtop{%
      \m@th % no extra space around math
      \offinterlineskip % disable glue between lines
      \ialign{%
        \hfil##\hfil\cr
        \rightarrowfill\cr
        \hphantom{$\scriptstyle\mskip8mu{#2}\mskip8mu$}\cr
        \vrule height0pt width 1.5em\cr
        $\scriptscriptstyle {#1}$\cr
      }%
    }%
  }\limits^{#2}%
}
\makeatother

\usepackage[disable]{todonotes}
\newcommand{\chana}[1]{\todo[color=green!30]{\small #1}}

\title{Reconfigurable Broadcast Networks and Asynchronous Shared-Memory Systems are Equivalent}
\author{A. R. Balasubramanian \qquad\qquad Chana Weil-Kennedy
\institute{Technical University of Munich \\ Munich, Germany \thanks{This project has received funding from the European Research Council (ERC) under the European Union's Horizon 2020 research and innovation programme under grant agreement No 787367 (PaVeS).}}
\email{bala.ayikudi@tum.de \quad\qquad chana.weilkennedy@in.tum.de}
}
\def\titlerunning{RBN and ASMS are Equivalent}
\def\authorrunning{A. R. Balasubramanian, C. Weil-Kennedy}

\begin{document}
\maketitle

\begin{abstract}
We show the equivalence of two distributed computing models, namely
reconfigurable broadcast networks (RBN) and asynchronous shared-memory systems (ASMS),
that were introduced independently.
%what are our models
Both RBN and ASMS are systems in which a collection of anonymous, finite-state processes
run the same protocol. In RBN, the processes communicate by selective broadcast: 
a process can broadcast a message which is received by all of its neighbors, and the set
of neighbors of a process can change arbitrarily over time.
In ASMS, the processes communicate by shared memory: a process can either write to or read from 
a shared register.
%what is a simulation
Our main result is that RBN and ASMS can \emph{simulate} each other, i.e. they are equivalent with respect
to parameterized reachability, where we are given two (possibly infinite) sets of configurations $\cube$ and $\cube'$ 
defined by upper and lower bounds on the number of processes in each state and we would like to decide if some configuration in $\cube$ can reach
some configuration in $\cube'$.
%We show that RBN can simulate ASMS:
%Given an RBN instance $P$, we can construct an ASMS instance $P'$ and a bijection $f$ from configurations of $P$ to 
%certain "good" configurations of $P'$ satisfying the following property:
%Given two sets of configurations $C,D$ 
%defined by upper and lower bounds on the number of processes in each state, 
%C$ can reach $D$ in the RBN $P$ if and only if 
%f(C)$ can reach $f(D)$ in the ASMS $P'$.
%Additionally, ASMS can simulate RBN.
%consequences
Using this simulation equivalence, we transfer results of RBN to ASMS and vice versa.
Finally, we show that RBN and ASMS can simulate a third distributed model
called immediate observation (IO) nets. Moreover, for a slightly stronger notion of simulation (which is satisfied by all the simulations 
given in this paper), we show that IO nets cannot simulate RBN.
\end{abstract}

%%%%%%%%%%%%%%%%%%%%%%%%%%%%%
\section{Introduction}
%\balain{Move stuff from 4.1 and 4.2 to the appendix; Clean appendix}
%\chanain{ shorten}

In this paper, we consider three models of distributed computation, one in which
communication happens by (selective) broadcasts, another in which communication happens by means of a shared memory, and finally one in which
communication happens by observation. We first expand a bit more on these models, then describe our main results and finally derive some
consequences from these results.

The first model that we consider is \emph{reconfigurable broadcast networks} (RBN)\cite{AdHocNetworks, FSTTCS12}.
In this model, we have a collection of anonymous, finite-state processes executing the same protocol. Further, every process has a set of neighbors.
At each step, a process can broadcast a message which is then received by all of the processes in its neighborhood. 
The neighborhood topology is reconfigurable, meaning that the set of neighbors of a process can change arbitrarily between two steps.
Parameterized verification of RBN
aims to prove that a property is correct, irrespective
of the number of participating processes.
Dually, it attempts to find an execution of some population of processes
for which a property is violated.  Within this context, the complexity of different variants of (parameterized) reachability and repeated coverability have been studied for RBN \cite{AdHocNetworks, FSTTCS12, Liveness}. 
Moreover, many extensions of RBN with clocks, registers and probabilities have been proposed and studied, mainly within the perspective of parameterized verification \cite{register,prob,probtime}.

The second model that we consider is a formal model of \emph{asynchronous shared-memory systems} (ASMS)\cite{ModelCheckingSMS, ParamLiveness, ICALPPatricia}.
In this model, we have a collection of anonymous, finite-state processes executing the same protocol, and a single register which all processes can access 
to perform a read/write operation. The set of values that can be stored in this register is finite. 
No locks onto the register are allowed and so no process can perform a sequence of atomic operations whilst preventing other processes from accessing the register.
Similar to RBN, major questions of interest in ASMS are those pertaining to parameterized verification, i.e. finding bad executions over
some population of processes.
The complexity of some (parameterized) reachability and model-checking questions for ASMS have been explored in a series of papers
\cite{ModelCheckingSMS, JACM16, ParamLiveness}. Further extensions of this model with leaders, stacks, etc. have also been studied \cite{JACM16, FineGrained, SafetyAlmostAlways, ModelCheckingPushdown}. Finally, \cite{ICALPPatricia} considers ASMS in the setting in which a stochastic scheduler
picks a process (uniformly at random) at each step to be executed, and under this setting studies the question of 
whether a given state can be reached by some process almost-surely, i.e., with probability 1. 

The third model that we consider is \emph{immediate observation Petri nets} (IO nets) \cite{EsparzaGMW18,EsparzaRW19}, 
which were introduced with motivations from the field of population protocols \cite{Comp-Power-Pop-Prot,First-Pop-Prot}.
Intuitively, in this model, we have a collection of anonymous, finite-state processes executing the same protocol. The only communication allowed
between processes is \emph{observation}, i.e., a process can only observe that another process is at some other state, and based on this 
observation can execute a step. The process being observed cannot detect if some process is observing it. Motivated by application to population protocols,
the authors of \cite{EsparzaGMW18, EsparzaRW19} study parameterized reachability questions for IO nets.

In this paper, we show that RBN and ASMS can \emph{simulate} each other, with respect to (parameterized) reachability.
Roughly speaking, we show that any instance of a parameterized reachability question for RBN can be efficiently translated to an instance
of parameterized reachability for ASMS and vice versa. More specifically, we consider the question of \emph{cube-reachability}. 
In the cube-reachability question, we are given an instance of a model (which can be either an RBN, an ASMS or an IO net) and two sets of configurations $\cube,\cube'$,
each of them defined by lower and upper bounds on the number of processes in each state.  (The upper bounds on some states might be $\infty$, which means
that we allow arbitrary number of processes in that state). We would then like to decide if there is a configuration in $\cube$ which can 
reach a configuration in $\cube'$. As we shall explain in the next section, the cube-reachability question covers parameterized reachability and coverability problems,
%correctness problem for IO nets,
parameterized reachability problems with leaders, and allows for a uniform transfer of results between the models that we study in this paper.

Our main result is that the cube-reachability questions for RBN and ASMS are polynomial-time equivalent to each other. This result, along with the constructions
achieving this result, enable us to translate various parameterized reachability results from RBN to ASMS and vice versa. First, we show that a special case of cube-reachability, called unbounded initial cube reachability,
is \textsf{PSPACE}-complete for ASMS, by using our reduction and already existing similar results on RBN. Then, we introduce the model of RBN-leader protocols and 
use already existing results on ASMS-leader protocols to prove that the RBN-leader reachability problem is \textsf{NP}-complete. Finally, we show that
the almost-sure coverability problem for RBN is decidable in \textsf{EXPSPACE} by translating the analogous result for ASMS~\cite{ICALPPatricia}.

Additionally, we show that the cube-reachability problem for IO nets reduces to the cube-reachability problem for RBN, leading to a transfer of results from RBN to IO nets. For the other direction, we actually show an impossibility result.
We define a stronger form of reduction for the cube-reachability problem and we notice that the 
reductions given in this paper all satisfy this stronger property. Then, using results from IO net theory, 
we show that there can be no reduction from the cube-reachability problem for RBN to the cube-reachability problem for IO nets which satisfies this stronger property.
We leave open the problem of whether there can exist other reductions from RBN to IO nets.

The rest of the paper is organized as follows: In Section 2, we present some preliminary definitions and notations, then in Section 3, we describe RBN and ASMS.
Section 4 proves our main result that RBN and ASMS can simulate each other. Section 5 presents some transfer of results between RBN and ASMS. 
In Section 6, we introduce IO nets, show that they can be simulated by RBN, and prove that the other direction is not true for a stronger form of simulation.
For space reasons, all missing proofs are relegated to the full version of this paper~\cite{LongVersion}.

\section{Preliminaries}

%!TEX root = main.tex

%same as PN19 up to small changes
\paragraph*{Multisets.}
A \emph{multiset} on a finite set \(E\) is a mapping \(C \colon E \rightarrow \N\), i.e. for any $e\in E$, \(C(e)\) denotes the number of occurrences of element \(e\) in \(C\).
We let $\mathbb{M}(E)$ denote the set of all multisets on $E$.
Let $\multiset{e_1,\ldots,e_n}$ denote the multiset $C$ such that $C(e)=|\{j\mid e_j=e\}|$.
We sometimes write multisets using set-like notation. 
For example, $\multiset{2 \cdot a,b}$ and $\multiset{a,a,b}$ denote the same multiset.
Given $e \in E$, we denote by $\vec{e}$ the multiset consisting of one occurrence of element
$e$, that is $\multiset{e}$. 
Operations on \(\N\) like addition or comparison are extended to multisets by defining them component wise on each element of \(E\).
Subtraction is allowed as long as each component stays non-negative.
Given a multiset $C$ on $E$  and a multiset $C'$ on $E'$ 
such that $E \cap E'=\emptyset$, we denote by 
$C \cdot C'$ the multiset on $E \cup E'$ equal
to $C$ on $E$ and  to $C'$ on $E'$. 
We call $|C| \defeq\sum_{e\in E} C(e)$ the \emph{size} of $C$, and $\support{C} \defeq \{ e \mid C(e)>0 \}$ the \emph{support} of $C$. 
Given \(E'\subseteq E\) define \(C(E')\defeq\sum_{e\in E'} C(e)\).
%Given a total order $e_1 \prec e_2 \prec \cdots \prec e_n$ on $E$, a multiset $C$ can be 
%equivalently represented by the vector $(C(e_1), \ldots, C(e_n))\in \N^n$. 
%In the rest of the paper, we assume a total order on the (finite) set of states of our systems.
%Given a set of states $Q$, a \emph{configuration} is a multiset on $Q$. \balain{Do we need these last three lines? The reason I ask is, configurations of an ASMS are strictly speaking,
%	not just multisets. Hence when we also use configurations there,  it might be confusing.}

\paragraph*{Cubes.}
Given a finite set $Q$, a \emph{cube} $\cube$ is a subset of $\mathbb{M}(Q)$ described 
by a lower bound $L \colon Q \rightarrow \N$ 
and an upper bound $U \colon Q \rightarrow \N \cup \{\infty\}$ 
such that $\cube = \{C : L \le C \le U\}$.
Abusing notation, we identify the set $\cube$ with the pair $(L,U)$.
All the results in this paper are true irrespective of whether the constants
are encoded in unary or binary.

\paragraph*{Reachability.}
Let $\mathcal{T} = (S,\rightarrow)$ be a transition system where $S$ is a set of \emph{configurations} and $\rightarrow$ is a binary relation on $S$
called the transition (or) step
relation.
%Let $\mathcal{T}$ be a system with configurations $C$, equipped with a step relation,  noted $\rightarrow$, between configurations.
Given configurations $C$ and $C'$, 
we say $C'$ is \emph{reachable} from $C$ if $C \trans{*} C'$
, where $\trans{*}$ denotes the reflexive-transitive closure of the step relation.
Let $\cSet$ be a set of configurations. 
The \emph{predecessor set} of $\cSet$ is 
$\pre^*_\mathcal{T}(\cSet) \defeq \{ C' | \exists C \in \cSet \, . \, C' \xrightarrow{*} C \}$, and the \emph{successor set} of $\cSet$ is
$\post^*_\mathcal{T}(\cSet) \defeq \{ C | \exists C' \in \cSet \, . \, C' \xrightarrow{*} C \}$.
The \emph{immediate predecessor set} of $\cSet$ is 
$\pre_\mathcal{T}(\cSet) \defeq \{ C' | \exists C \in \cSet \, . \, C' \rightarrow C \}$, and the \emph{immediate successor set} of $\cSet$ is
$\post_\mathcal{T}(\cSet) \defeq \{ C | \exists C' \in \cSet \, . \, C' \rightarrow C \}$.
When it is clear from the context, we will drop the $\mathcal{T}$ subscript.
The \emph{reachability} problem consists of deciding, 
given a system $\mathcal{T}$ and configurations $C,C'$, 
whether $C'$ is reachable from $C$ in $\mathcal{T}$.

\paragraph*{Cube reachability.}

If $\mathcal{T}$ is a transition system whose set of configurations is the set of all multisets on
a finite set $Q$, then the reachability problem can be generalized to the \emph{cube-reachability} problem which consists of deciding, given $\mathcal{T}$ and two cubes $\cube, \cube'$ over $Q$,
whether there exists configurations  $C \in \cube$ and $C' \in \cube'$ such that $C'$ is reachable from $C$ in $\mathcal{T}$.
If this is the case, we say $\cube'$ is reachable from $\cube$.

As mentioned before, the cube-reachability problem generalizes the reachability problem. It also generalizes the coverability problem : Given a configuration $C$ and a state $q \in Q$,
decide if there exists $C'$ such that $C \act{*} C'$ and $C'(q) \ge 1$. It can also talk about \emph{parameterized reachability} problems, for e.g., 
given two finite sets of states $I$ and $F$, do there exist configurations $C$ and $C'$ such that
$\support{C} \subseteq I, \support{C'} \subseteq F$ and $C \act{*} C'$.
Further, the cube-reachability problem is important in the model of immediate observation Petri nets (IO nets). This model was introduced to study immediate observation population protocols \cite{EsparzaGMW18, EsparzaRW19}, and the correctness problem for these protocols is solved using  cube-reachability in IO nets.% \bala{Check if this is okay} \chana{changed  it}
%Further, the cube-reachability problem is important in the model of immediate-observation (IO) protocols \cite{EsparzaGMW18, EsparzaRW19}, where it is used to solve their correctness problem. (We will define a similar model to IO protocols called IO nets in Section~\ref{sec:IO}  \bala{Check if this is okay} and relate it with the other models that we study).
Additionally, as we will see in Section~\ref{subsec:leader}, the cube-reachability problem is a generalization of the so-called \emph{leader reachability problem} and allows for an elegant way to transfer results between the models that we study in this paper.

%A formal finite union of cubes $\bigcup_{i=1}^m (L_i,U_i)$ is called a \emph{counting constraint}
%and the set of configurations $\bigcup_{i=1}^m \cube_i$ it describes is called a \emph{counting set}.
%formal enough?
%We write $\sem{\cC}$ the counting set described by counting constraint $\cC$.
%Notice that two different counting constraint may describe the same counting set.
%For example, let $Q=\set{q}$ and let $(L,U)=(1,3)$, $(L',U')=(2,4)$, $(L'',U'')=(1,4)$. 
%The counting constraints $(L,U)\cup(L',U')$ and $(L'',U'')$ define the same counting set.
%the following sentence is directly from PN2019, sorry plagiarism gods
%It is easy to show (see also \cite{EsparzaGMW18}) that counting constraints and counting sets are closed under Boolean operations.
%The \emph{cube-reachability} problem consists of deciding, given a system $\mathcal{T}$ and cubes $\cube,\cube'$, whether there exist configurations 
%$C \in \cube$ and $C' \in \cube'$ such that $C'$ is reachable from $C$ in $\mathcal{T}$.
%If this is the case, we say $\cube'$ is reachable from $\cube$.

%%%%%%%%%%%%%%%%%%%%%%%%%%%%%
\section{Two Models}

\subsection{Reconfigurable Broadcast Networks}
%!TEX root = main.tex

%Notations for RBN follow [Bertrand, Bouyer, Majumdar, CONCUR 2019] and [Delzanno, Sangnier, Traverso, Zavattaro, FSTTCS 2012].
%We call \emph{configuration} of an RBN the vector counting the number of nodes in each state, for a given ordering of the states.
%In the RBN papers, a configuration is an irreflexive undirected graph with nodes labeled by states and the edges representing the neighbour relation.
%We can forget this neighbour information without loss of generality when considering reachability, 
%since the neighbour relation can be reconfigured at the start of an RBN step.
%A step in an RBN is one such reconfiguration, then one send transition from a node $u$ and then a receive transition from each of the neighbours of $u$.
%A step in an IO net is the firing of a transition of the form $p \trans{q} r$, where a token in $p$ moves to $r$ after observing a token in $q$.

Reconfigurable broadcast networks (RBN)~\cite{AdHocNetworks,FSTTCS12} are networks comprising an arbitrary number of finite-state, anonymous processes and a communication topology specifying the presence or absence of communication links
between different processes. During a step, a process can broadcast a message which is immediately received by all of its neighbours. 
The process and its neighbours then update their states according to a transition relation. Before each such broadcast step, the communication topology can reconfigure in an arbitrary manner.
Since our main focus in this paper is regarding reachability in this model, we can forget the communication topology and simply define the semantics of an RBN directly in 
terms of collections of processes. 
%We can forget this neighbour information without loss of generality when considering reachability, and so we do not present it in this paper.

\begin{definition}
\label{def:rbn}
A \emph{reconfigurable broadcast network} is a tuple 
$\RBN = (Q, \Sigma,\delta)$ 
where $Q$ is a finite set of states,
$\Sigma$ is a finite alphabet 
and $\delta \subseteq Q \times \set{!a,?a \ | \ a \in \Sigma} \times Q$ is the transition relation.
\end{definition}

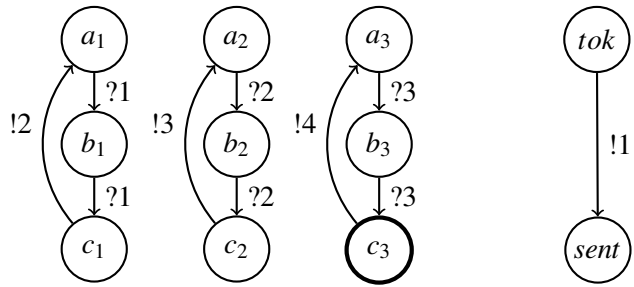
\begin{figure}
%!TEX root = main.tex
\begin{center}
    \begin{tikzpicture}[->, thick]
      \node[place] (a1) {$a_1$};
      \node[place] (b1) [below =0.5 of a1] {$b_1$};
      \node[place] (c1) [below =0.5 of b1] {$c_1$};
      \node[place] (a2) [right =of a1] {$a_2$};
      \node[place] (b2) [below =0.5 of a2] {$b_2$};
      \node[place] (c2) [below =0.5 of b2] {$c_2$};
      \node[place] (a3) [right =of a2] {$a_3$};
      \node[place] (b3) [below =0.5 of a3] {$b_3$};
      \node[place, ultra thick] (c3) [below =0.5 of b3] {$c_3$};
      \node[place] (tok) [right =2 of a3] {$tok$};
      \node[place] (sent) [right =2 of c3] {$sent$};
      
      \path[->]
      (tok) edge node[right] {$!1$} (sent)
      (a1) edge node[right] {$?1$} (b1)
      (b1) edge node[right] {$?1$} (c1)
      (c1) edge[bend left=40] node[above left] {$!2$} (a1)
      (a2) edge node[right] {$?2$} (b2)
      (b2) edge node[right] {$?2$} (c2)
      (c2) edge[bend left=40] node[above left] {$!3$} (a2)
      (a3) edge node[right] {$?3$} (b3)
      (b3) edge node[right] {$?3$} (c3)
      (c3) edge[bend left=40] node[above left] {$!4$} (a3)
      ;
      %\draw[->,bend left=20] (c1) to (a1);
    \end{tikzpicture}
\end{center}
\caption{An RBN simulating a counter to $2^3$.}
\label{fig:rbn}
\end{figure}

If $(p,!a,q)$ (resp. $(p,?a,q)$) is a transition in $\delta$, we will denote it by $p \act{!a} q$ (resp. $p \act{?a} q$).
%We denote by $p \trans{!a} q$ and $p \trans{?a} q$, the transitions $(p,!a,q)$ and $(p,?a,q) \in \delta$.
A \emph{configuration} $C$ of an RBN $\RBN$ is a multiset over $Q$, which
intuitively counts the number of processes in each state. 
Given a letter $a\in \Sigma$ and two configurations $C$ and $C'$
we say  that there is a \emph{step} $C \trans{a} C'$
if there exists a multiset $\multiset{t, t_1, \ldots, t_k}$ of $\delta$ for some $k\ge 0$
satisfying the following: $t=p \trans{!a} q$, each $t_i =p_i \trans{?a} q_i$,
$C \ge \vec{p} + \sum_i \vec{p_i}$, and $C' = C - \vec{p} - \sum_i \vec{p_i} + \vec{q} + \sum_i \vec{q_i}$. 
We sometimes write this as $C \trans{t+t_1,\ldots, t_n} C'$, 
and intuitively it means that a process at the state $p$ broadcasts the message $a$ and moves to $q$,
and for each $1 \le i \le k$, there is a process at the state $p_i$ which receives this message and moves to $q_i$.
We denote 
%by $\trans{t+t_1 \ldots t_k}$ such a step, and 
by $\trans{*}$ the reflexive and transitive closure of the step relation. 
A \emph{run} is then a sequence of steps.

\begin{example}
\label{ex:rbn}
Consider the RBN of Figure \ref{fig:rbn}, 
with set of states $\set{tok, sent} \cup \set{a_i,b_i,c_i | 1 \le i \le 3}$.
It is inspired by a similar example described in Section 5.1 of \cite{ICALPPatricia}.
Let $\cube_0$ be the cube which puts exactly
one process in each $a_i$, an arbitrary number of processes in $tok$
and $0$ processes elsewhere. 
That is, $\cube_0=(L,U)$ such that $L(a_i)=U(a_i)=1$ for all $i$,
$L(tok)=0$ and $U(tok)=\infty$, and $L(q)=U(q)=0$ for all other states $q$.
%\bala{Do you think this informal description of cubes is fine? We haven't really talked about the intuitive definition of cubes in preliminaries?}
Let $\cube_f$ be the cube which puts at least one process in $c_3$ and an arbitrary number elsewhere.
Suppose some configuration in $\cube_0$ reaches some configuration in $\cube_f$.
By construction, for a process to reach $c_3$ it must start in $a_3$ and receive $3$ twice.
For a process to broadcast $3$ it must start in $a_2$ and receive $2$ twice, and for a process to broadcast $2$ it must start in $a_1$ and receive $1$ twice.
So a run from a configuration of $\cube_0$ to a configuration of $\cube_f$
must contain at least $2^3$ broadcasts of $1$.
Since the only way to broadcast $1$ is for a process
to go from $tok$ to $sent$,
 there must be at least $2^3$ processes in $tok$ in 
 the initial configuration of $\cube_0$.
\end{example}
\subsection{Asynchronous Shared-Memory Systems}
%!TEX root = main.tex
Asynchronous shared-memory systems (ASMS)~\cite{JACM16,ModelCheckingSMS}
consist of an arbitrary number of finite-state, anonymous processes.
%processes running the same (finite-state) template.
These processes can communicate with each other by means of a single shared
register, to which they can either write a value or from which they can read a value.

\begin{definition}
	An asynchronous shared-memory system (ASMS) is a tuple $\ASMS = (Q,\Sigma,\delta)$ where
	$Q$ is a finite set of states, $\Sigma$ is a finite alphabet,
	and $\delta \subseteq Q \times \{R,W\} \times \Sigma \times Q$ is the set of transitions.
	Here $R$ stands for \emph{read}, and
	$W$ stands for \emph{write}.
\end{definition}

We use $p \trans{R(d)} q$ (resp. $p \trans{W(d)} q$) to denote
that $(p,R,d,q) \in \delta$ (resp. $(p,W,d,q) \in \delta$). 
The semantics of an ASMS is given by means of \emph{configurations}.
A configuration $C$ of an ASMS is a multiset over $Q \cup \Sigma$
such that $\sum_{d \in \Sigma} C(d) = 1$, i.e., $C$ contains
exactly one element from the set $\Sigma$. 
Hence, we sometimes denote a configuration $C$ as $(M,d)$
where $M$ is a multiset over $Q$ (which counts the number of processes
in each state) and $d \in \Sigma$ (which denotes the content of the shared register).
%A configuration $C$ of an ASMS is a tuple $(M,d)$ where
%$M$ is a multiset over $Q$ (which counts the number of processes
%in each state) and $d \in \Sigma$ (which denotes the content of the shared
%register). 
The value $d$ will be denoted by $\data(C)$.
%Further, we will overload the operators defined for multisets to
%accommodate them for configurations, for example given $C = (M,d)$,
%we will use the notation $C(q)$ to denote $M(q)$. 

A \emph{step} between configurations $C = (M,d)$ and $C' = (M',d')$ exists
if there is $t = (p,\oper,d'',q) \in \delta$ such that
$M(p) > 0$, $M' = M - \vec{p} + \vec{q}$
and either $\oper = R$ and $d = d' = d''$ or 
$\oper = W$ and $d' = d''$. If such a step exists,
we denote it by $C \trans{t} C'$ and we let $\trans{*}$ denote
the reflexive transitive closure of the step relation. 
A \emph{run} is then a sequence of steps.
Given a sequence of transitions $\sigma = t_1,\dots,t_n$, we 
sometimes use $C \act{\sigma} C'$ to denote that there is 
a run of the form $C \act{t_1} C_1 \act{t_2} \dots C_{n-1} \act{t_n} C'$.

A cube $\cube = (L,U)$ of an ASMS $\ASMS = (Q,\Sigma,\delta)$ is defined to be a cube over $Q \cup \Sigma$
satisfying the following property : There exists $d \in \Sigma$ such that $L(d) = U(d) = 1$ and $L(d') = U(d') = 0$ for every other $d'$.
Hence, we sometimes denote a cube $\cube$ as $(L,U,d)$ where $(L,U)$ is a cube over $Q$ and $d \in \Sigma$.
Membership of a configuration $C$ in a cube $\cube$ is then defined in a straightforward manner.
%Because of the presence of the external register, we need a slightly different notion of cubes and cube reachability for ASMS.
%A cube of an ASMS $\ASMS$ is simply a cube over the states of $\ASMS$.
%Given a cube $\cube$ of $\ASMS$, we say that a configuration $C = (M,d) \in \cube$
%if $M \in \cube$ and $d = \#$. Notice that the data value of $C$ must be the distinguished letter $\#$.
The cube-reachability problem for ASMS is then to decide, given $\ASMS$ and two cubes $\cube, \cube'$
whether $\cube$ can reach $\cube'$, i.e., whether there are configurations $C \in \cube, C' \in \cube'$
such that $C \act{*} C'$. 

\begin{figure}
	\begin{center}
    \begin{tikzpicture}[->, thick]
      \node[place] (a1) {$a_1$};
	  \node[place] (a2) [right =of a1] {$a_2$};
	  \node[place] (a3) [right =of a2] {$a_3$};
      \node[place] (a4) [right =of a3] {$a_4$};

      \node[place] (b3) [below =0.75 of a1] {$b_3$};      
      \node[place] (b2) [left =of b3] {$b_2$};
      \node[place] (b1) [left =of b2] {$b_1$};
      
      \node[place] (c1) [below =0.75 of a4] {$c_1$};  
      \node[place] (c2) [right = of c1] {$c_2$};
	  \node[place] (c3) [right = of c2] {$c_3$};
      %\node[place] (b3) [below =0.5 of a3] {$b_3$};
      %\node[place, ultra thick] (c3) [below =0.5 of b3] {$c_3$};

      %\node[place] (sent) [right =2 of c3] {$sent$};
      
      \path[->]
      (a1) edge[bend left = 40] node[above] {$W(1)$} (a2)
      (a1) edge[bend right = 40] node[below] {$W(2)$} (a2)
      (a2) edge node[above] {$R(3)$} (a3)
      (a3) edge node[above] {$R(4)$} (a4)
      %(a4) edge[loop above] node[above] {$W(\#)$} (a4)
      
      (b1) edge node[above] {$R(1)$} (b2)
      (b2) edge node[above] {$W(3)$} (b3)
      
      (c1) edge node[above] {$R(2)$} (c2)
      (c2) edge node[above] {$W(4)$} (c3)
      %(c3) edge[bend left=40] node[above left] {$!4$} (a3)
      ;
    \end{tikzpicture}
\end{center}
	\caption{An example of an ASMS}
	\label{fig:asms}
\end{figure}
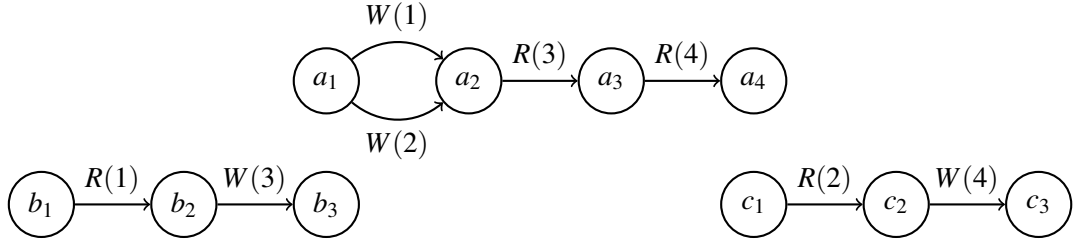

\begin{example}
\label{ex:asms}
Consider the ASMS of Figure~\ref{fig:asms} where the alphabet is $\{\#,1,2,3,4\}$. Let $\cube$ be the cube 
which puts exactly one process in $a_1$, arbitrary number of processes in $b_1$ and $c_1$ and exactly 0 processes elsewhere.
Let $\cube'$ be the cube where $\cube'$ puts at least one process in $a_4$ and arbitrary number of processes elsewhere.
It can be verified that the cube $\cube$ cannot reach $\cube'$ for the following reason: Since there
is only one process in $a_1$ in $\cube$, it follows that this process can either write 1 or 2, but not both.
Hence, either processes from $b_1$ can move into $b_2$ to write 3 or processes from $c_1$ can move into $c_2$ to write 4, but 
both cannot happen. It then follows that it is impossible to read both 3 and 4, and so the state $a_4$ cannot be reached.
\end{example}

%%%%%%%%%%%%%%%%%%%%%%%%%%%%%
\section{RBN and ASMS are Cube-Reachability Equivalent}\label{sec:simulations}

%!TEX root = main.tex
%For the purpose of this paper, a model is either an RBN, an ASMS or an IO (introduced in Section \ref{}).

Throughout this paper, whenever we talk about one model simulating another model, we mean that the cube-reachability problem for the second model
can be reduced in polynomial time to the cube-reachability problem for the first model.
In this section, we prove our main result that RBN and ASMS can \emph{simulate} each other.
%More precisely, we show that the cube-reachability problem for RBN is polynomial-time equivalent to the cube-reachability problem for ASMS.
As we will see in the next section, this simulation will allow us to transfer results from RBN to ASMS and vice versa.

%Given two models $A$ and $B$, 
%model $B$ \emph{simulates} model $A$
%if $B$ is of size polynomial in the size 
%\chana{add what the size of a model is}
%of $A$, and for every two cubes $\cube_A, \cube_A'$ of $A$ 
%we can construct in polynomial time two cubes $\cube_B, \cube_B'$  of $B$ 
%such that $\cube_A'$ is reachable from $\cube_A$
%if and only if $\cube_B'$ is reachable from $\cube_A'$.

\subsection{ASMS Simulate RBN}\label{subsec:ASMS-RBN}
%!TEX root = main.tex
\textbf{Construction} \
Let $\RBN = (Q_\RBN, \Sigma_\RBN,\delta_\RBN)$ be an RBN. 
We construct an ASMS that simulates $\RBN$.
The register value is used to store which message can be received, additional states are used to represent that a broadcast is in progress, and a fresh register value is written when the simulation of a broadcast is over.
\chana{high level intuition added}
%%OR THE FOLLOWING
%We construct an ASMS $\ASMS$ which simulates a broadcast of $\RBN$ in the following way: the register stores the sent message, processes read from it and move to a intermediary position, and a fresh message is written to  the register to end the simulation.
For every $a \in \Sigma_\RBN$, we let $\delta_\RBN^{!a}$ (resp. $\delta_\RBN^{?a}$) be the subset 
of the transitions in $\delta_\RBN$ that broadcast (resp. receive) the letter $a$.
Let $\ASMS = (Q_\ASMS,\Sigma_\ASMS,\delta_\ASMS)$ be the following ASMS: 
The set of states $Q_\ASMS$ is $Q_\RBN \cup I$ with $I=\set{[p,a,p'] : (p,?a,p') \in \delta \text{ or } (p,!a,p') \in \delta}$, where $I$ stands for intermediary.
The alphabet $\Sigma_\ASMS$ is $\Sigma_\RBN \cup \{\#\}$ where $\#$ is a letter which is not in  $\Sigma$.
The transition relation $\delta_\ASMS$ is such that 
for every $t = (q,!a,q') \in \delta_\RBN$ 
there are transitions $\hat{t} := q \trans{W(a)} [q,a,q']$ and $t^\# := [q,a,q'] \trans{W(\#)} q'$ in $\delta_\ASMS$, and
for every $t = (q,?a,q') \in \delta_\RBN$ 
there are transitions $\hat{t} := q \trans{R(a)} [q,a,q']$ and $t^\# := [q,a,q'] \trans{W(\#)} q'$ in $\delta_\ASMS$,
as represented in Figure \ref{fig:rbn-to-sms}. \chana{fig 3 added and mentioned}

\begin{figure}[h]
\begin{center}
    \begin{tikzpicture}[->, thick, node distance=1.25cm] 
      \node[place] (q) {$q$}; 
      \node[place,inner sep=1pt] (qaq') [right =of q] {$q,a,q'$};
      \node[place] (q') [right =of qaq'] {$q'$};
      \node[place] (p) [right =of q'] {$p$};
      \node[place,inner sep=1pt] (pap') [right =of p] {$p,a,p'$};
      \node[place] (p') [right =of pap'] {$p'$}; 
      
      \path[->]
      (q) edge node[above] {$W(a)$} (qaq')
      (p) edge node[above] {$R(a)$} (pap')
      (pap') edge node[above] {$W(\#)$} (p')
      (qaq') edge node[above] {$W(\#)$} (q')
      ;
    \end{tikzpicture}
\end{center}
\caption{Simulation in $\ASMS$ of transitions $q\trans{!a}q'$ and $p \trans{?a}p'$ of $\RBN$.}
\label{fig:rbn-to-sms}
\end{figure}

A configuration $(C,d)$ of $\ASMS$ is called \emph{good} if
$C(I)=0$ and $d=\#$. 
There is a natural bijection between configurations of $\RBN$ and good configurations of $\ASMS$. 
If $C$ is a configuration of $\RBN$, we will use $(\widehat{C},\#)$ to denote the corresponding good configuration of $\ASMS$.

\noindent \textbf{Correctness of construction} \ 
We now show that $C' \in \poststar_{\RBN}(C)$ iff $(\widehat{C'},\#) \in \poststar_{\ASMS}(\widehat{C},\#)$ for any configurations $C$ and $C'$ of $\RBN$.
%Suppose $C \trans{a} C'$ is a step in $\RBN$, 
%$(\widehat{C},\#) \trans{*} (\widehat{C'},\#)$ in $\ASMS$.
%which is realized by the multiset of transitions
%$\multiset{(q,!a,q'), (p_1,?a,p'_1), \ldots, (p_n,?a,p'_n)}$
%for some $n\ge 0$.
Suppose $C \trans{t + t_1,\dots,t_n} C'$ is a step in $\RBN$.
It is easy to see that we have a run in $\ASMS$ of the form $(\widehat{C},\#) \act{\hat{t},\hat{t_1},
\dots,\hat{t_n},t^\#,t_1^\#,\dots,t_n^\#} (\widehat{C'},\#)$.
%Notice that corresponding to the transition $t := (q,!a,q')$, we have a transition $t^W := q\act{W(a)}q'$ in $\ASMS$
%and corresponding to each $t_i := (p_i,?a,p_i')$ we have two transitions $t_i^R := p_i\act{R(a)}[p_i,a,p_i']$
%and $t_i^W := [p_i,a,p_i']\act{W(\#)}p_i'$.
%It is then easy to see that we have a run of the form
%$(\widehat{C},\#) \act{t^W} (C_0,a) \act{t_1^R} (C_1,a) \act{t_2^R} \dots \act{t_n^R} (C_n,a) \act{t_1^W} (C_1',\#) \dots \act{t_n^W} (\widehat{C'},\#)$.
%We have $(\widehat{C},\#) \trans{W(a)} \trans{R(a)^n} \trans{W(\#)^n} (\widehat{C'},\#)$ in $\ASMS$,
%where we write $\trans{D(a)}$ to symbolize a transition 
%of the form $p \trans{D(a)} p'$, for $D \in \set{W,R}, a \in \Sigma$.
Hence, if $C \trans{*} C'$ for some configurations $C,C'$ in $\RBN$,
then $(\widehat{C},\#) \trans{*} (\widehat{C'},\#)$ in $\ASMS$.

For the other direction, we first define the notion of a \emph{pseudo-step} between two good configurations of $\ASMS$. 
A run $(\widehat{C},\#) \act{\sigma} (\widehat{C'},\#)$ of $\ASMS$ 
is called a \emph{pseudo-step} if there exists $a \in \Sigma$
and transitions $t \in \delta_{\RBN}^{!a}$ and $t_1,\dots,t_n \in \delta_{\RBN}^{?a}$ 
such that $\sigma = \hat{t},\hat{t_1},\dots,\hat{t_n},t^\#,t_1^\#,\dots,t_n^\#$.
The intuition behind this notion is that if $(\widehat{C},\#) \act{\sigma} (\widehat{C'},\#)$ 
where $\sigma$ is a pseudo-step with $\sigma = \hat{t},\hat{t_1},\dots,\hat{t_n},t^\#,t_1^\#,\dots,t_n^\#$
then $C \act{t+t_1\dots t_n} C'$ is a step in $\RBN$. Hence, pseudo-steps of $\ASMS$ ``behave'' similarly to a single step in $\RBN$.

Now a run $(\widehat{C},\#) \act{\sigma} (\widehat{C'},\#)$ of $\ASMS$ is said to be in \emph{normal form} if it is either the empty run or if it can be decomposed into a sequence of \emph{pseudo-steps}.
%there exists $(\widehat{C_1},\#),\dots,(\widehat{C_k},\#)$ such that $(\widehat{C},\#) \act{\sigma} (\widehat{C'},\#)$ can be decomposed as $(\widehat{C},\#) = (\widehat{C_0},\#) \act{\sigma_1} (\widehat{C_1},\#) \act{\sigma_2} (\widehat{C_2},\#) \dots (\widehat{C_k},\#) \act{\sigma_{k+1}} (\widehat{C_{k+1}},\#) = (\widehat{C'},\#)$
%and each $(\widehat{C_i},\#) \act{\sigma_{i+1}} (\widehat{C_{i+1}},\#)$ is a pseudo-step.
Hence, it follows that if $(\widehat{C},\#) \act{\sigma} (\widehat{C'},\#)$
is a run in normal form then $C \act{*} C'$ in $\RBN$. 

The following lemma asserts that whenever there is a run between two good configurations of 
$\ASMS$, then there is also a run between those configurations in normal form. Hence, using this
lemma and the discussion in the previous paragraph, it follows that if $(\widehat{C},\#) \act{*} (\widehat{C'},\#)$ in $\ASMS$ then $C \act{*} C'$ in $\RBN$.

%For the other direction, we first assume that every run in $\ASMS$ 
%can be put in \emph{normal form} where a run $\sigma$ is said to be in normal form if 
%$\sigma = \sigma_1 \sigma_2 \dots \sigma_n$ for some $n \ge 0$, such that each
%$\sigma_i= W(a_i) R(a_i)^{k_i} W(\#)^{k_i}$ \chana{not really defined; but it should do}
%for some $k_i \ge 0$ and $a_i \in \Sigma$.
%Assuming this, it is easy to see that if $(\widehat{C},\#) \trans{\sigma}_\ASMS (\widehat{C'},\#)$, then
%$C \trans{\sigma'}_\RBN C'$ where $\sigma'=\sigma'_1 \ldots \sigma'_n$ 
%and $\sigma'_i$ is a step of the form $\trans{a_i}$ with $k_i$ processes receiving the broadcast.
%$\trans{!a_i \ + \  ({k_i} \times ?a_i)}$. 

%\noindent \textbf{The normal form lemma} \ All that remains to prove 
%is the following normal form lemma.
\begin{restatable}[Normal form lemma]{lemma}{LmNormalASMS}
    Suppose $(\widehat{C},\#) \act{\rho} (\widehat{C'},\#)$ is a run in $\ASMS$.
    Then there exists  $\sigma$ such that $(\widehat{C},\#) \act{\sigma} (\widehat{C'},\#)$
    is a run in normal form.
\end{restatable}

\begin{proof}[Proof sketch of normal form lemma]
Let $n$ be  the length of $\rho$. \chana{correcting induction start. is it ok that length of sequences is not defined?}
We proceed by induction on $n$. If $n = 0$, we are done.
Let $n > 0$ and $\rho = \rho_1,\dots,\rho_n$. Assume now that any run of length strictly less than $n$ can be put in normal form. By analysing the
structure of the transitions in $\ASMS$ and noticing that $\rho$ begins at a good configuration,
we can first show that $\rho_1$ must be of the form $q \act{W(a)} [q,a,q']$ for some $a \in \Sigma$. 
Then we consider two cases:

\textbf{Case 1: } Suppose there is no $i > 1$ such that $\rho_i$ is a transition which
writes a value $b \neq \#$.  Hence, every transition in $\rho_2,\dots,\rho_n$ either reads the value $a$ or writes $\#$ and so
there must be an index $2 \le j \le n$ such that every transition in $\rho_2,\dots,\rho_{j-1}$
reads $a$ and every transition in $\rho_j, \dots, \rho_n$ writes $\#$. Now, by analysing
the transitions going in and out of the subset $I$ and noticing that the run
begins and ends at good configurations, we can show that
$(\widehat{C},\#) \act{\rho} (\widehat{C'},\#)$ must be a pseudo-step.

\textbf{Case 2: } Suppose there is $i > 1$ such that $\rho_i$ is a transition which writes
a value $b \neq \#$. By the same argument as before, it is easy to see that 
there must exist $2 \le j \le i-1$ such that every transition in $\rho_2,\dots,\rho_{j-1}$
reads $a$ and every transition in $\rho_j, \dots, \rho_{i-1}$ writes $\#$. Let $Z$
be the configuration reached after $\rho_{i-1}$.
Let $M = Z(I)$, i.e., $M$ is the multiset of processes at the configuration $Z$
which are in some intermediary state. Since the only way out of the set $I$ is to
write $\#$ onto the register, if $M = \multiset{[p_1,a,p_1'],\dots,[p_k,a,p_k']}$ then
there must exist $i_1,\dots,i_k > i$ such that each $\rho_{i_l}$ is $[p_l,a,p_l'] \act{W(\#)} p_l'$. 
We can then rearrange the run by first following $\rho$ up till $\rho_{i-1}$, then ``preponing'' the
transitions $\rho_{i_1},\dots,\rho_{i_k}$ and then firing the rest of $\rho$ to reach $(\widehat{C'},\#)$.
With this rearrangement, the run up till $\rho_{i_k}$ becomes a pseudo-step
and so we can apply induction hypothesis on the rest of the run.
\end{proof}

\noindent \textbf{The reduction} \ With this construction, we can now simulate RBN by ASMS as follows: Let $\RBN$ be an RBN with states $Q_\RBN$ and let $\cube_1 = (L_1,U_1),\cube_1' = (L_1',U_1')$ be two cubes of $\RBN$. Construct the ASMS $\ASMS$ as described above.
Then construct the following two cubes $\cube_2 = (L_2,U_2,\#), \cube_2' = (L_2',U_2',\#)$ of $\ASMS$:
$L_2(q), U_2(q), L_2'(q)$ and $U_2'(q)$ are respectively equal to $L_1(q), U_1(q), L_1'(q)$ and $U_1'(q)$ if $q$ is a state of $\RBN$.
If $q$ is in $I$, then $L_2(q) = U_2(q) = L_2'(q) = U_2'(q) = 0$.
It is then easy to see that $C \in \cube_1$ (resp. $\cube_1'$) iff $\hat{C} \in \cube_2$ (resp. $\cube_2'$). Hence, by correctness of our construction, it follows that $\cube_1$ can reach $\cube_1'$ iff $\cube_2$ can reach $\cube_2'$.

\subsection{RBN Simulate ASMS}\label{subsec:RBN-ASMS}
%!TEX root = main.tex

%Let $\ASMS = (Q_{\ASMS},\init_{\ASMS},\Sigma,\#,\delta_{\ASMS})$ be an ASMS. 
%We will construct an RBN so that one process simulates the register of $\ASMS$ and
%all the other processes behave like processes of $\ASMS$.
%Construct the RBN $\RBN = (Q_{\RBN},\init_{\RBN},\Sigma_{\RBN},\delta_{\RBN})$ as follows:
%The set of states $Q_{\RBN}$ will be $Q_{\ASMS} \cup \Sigma \cup \{p[a]q: p \trans{W(a)} q \in \delta_{\ASMS}\} 
%\cup \{\overline{a} : a \in \Sigma\}$. The set $\{\overline{a} : a \in \Sigma\}$ will be denoted by $\overline{\Sigma}$.
%The set of initial states will be $\init_{\ASMS} \cup \{\#\}$.
%The alphabet $\Sigma_{\RBN}$ will be $\{Read_a, Change_a, Ack_a : a \in \Sigma\}$.

\textbf{Construction} \ Let $\ASMS = (Q_{\ASMS},\Sigma,\delta_{\ASMS})$ be an ASMS. 
We construct an RBN $\RBN$ where one agent  acts like the register of $\ASMS$ and
all the other agents behave like agents of $\ASMS$.
Let $\RBN = (Q_{\RBN},\Sigma_{\RBN},\delta_{\RBN})$ be an RBN defined as follows:
The set of states $Q_{\RBN}$ is comprised of two parts.
The first part consists of the set $Q_{\ASMS} \cup \{[p,a,q]: p \trans{W(a)} q \in \delta_{\ASMS}\}$,
which will  intuitively be used to simulate the processes of $\ASMS$.
The second part consists of the set $\Sigma \cup \{\overline{a} : a \in \Sigma\}$
which will  intuitively be used to simulate the register of $\ASMS$.
%The third part consists of a single state $\sink$, which, as the name suggests, is a sink state. 
The set $\{\overline{a} : a \in \Sigma\}$ is denoted by $\overline{\Sigma}$.
The alphabet $\Sigma_{\RBN}$ is $\{Read_a, Ch_a, Ack_a : a \in \Sigma\}$.

Before describing the transition relation $\delta_{\RBN}$ we set up some notation:
A \emph{good} configuration of $\RBN$ is a configuration $C$ 
such that $\sum_{a \in \Sigma} C(a) = 1$ and
$C(p) = 0$ if $p \notin Q_{\ASMS} \cup \Sigma$.
Intuitively, in a good configuration, there is one process which stores
the value of the register of $\ASMS$ and all the
other processes are in some state of $Q_{\ASMS}$.
Notice that there is a natural bijection between configurations of $\ASMS$ and good configurations of $\RBN$.
If $C$ is a configuration of $\ASMS$, we will use $\widehat{C}$ to denote the corresponding good configuration of $\RBN$.

\begin{figure}
\begin{center}
    \begin{tikzpicture}[->, thick, node distance=1.25cm] 
      \node[place] (q) {$q$}; 
      \node[place,inner sep=1pt] (qaq') [right =of q] {$q,a,q'$};
      \node[place] (q') [right =of qaq'] {$q'$};
      \node[place] (d) [right =of q'] {$d$};
      \node[place] (abar) [right =of d] {$\overline{a}$};
      \node[place] (a) [right =of abar] {$a$}; 
      \node[place] (p) [right =of a] {$p$};
      \node[place] (p') [right =of p] {$p'$}; 
      
      \path[->]
      (q) edge node[above] {$?Ch_a$} (qaq')
      (qaq') edge node[above] {$!Ack_a$} (q')
      (p) edge node[above] {$?Read_a$} (p')
      (d) edge node[above] {$!Ch_a$} (abar)
      (a) edge [loop above] node {$!Read_a$} (a)
      (abar) edge node[above] {$?Ack_a$} (a)
      ;
      
      %\node[] () [above= -1pt of q] {$q$};
    \end{tikzpicture}
\end{center}
\caption{Simulation in $\RBN$ of transitions $q\trans{W(a)}q'$ and $p \trans{R(a)}p'$ of $\ASMS$.}
\label{fig:sms-to-rbn}
\end{figure}
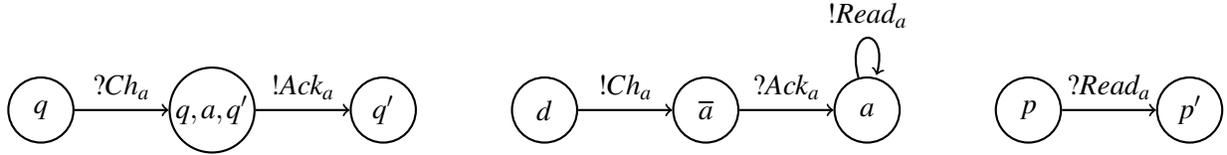

Now, the transition relation $\delta_{\RBN}$ is constructed so that the following invariant is satisfied:
For any configurations $C$ and $C'$ of $\ASMS$, $C' \in \poststar_{\ASMS}(C)$ iff $\widehat{C'} \in \poststar_{\RBN}(\widehat{C})$.
%With this invariant in mind, let us construct $\delta_{\RBN}$. 
%In the following, we adopt the convention that if 
%we do not specify what happens when a message $d$ is received at a state $q$, then there is a transition of the form $(q,?d,\sink)$
%in $\delta_\RBN$. 

\begin{itemize}
	\item Suppose $t = p \trans{R(a)} q$ is a transition in $\ASMS$. Correspondingly, we have 
	two transitions $a \trans{!Read_a} a$ and $p \trans{?Read_a} q$ in $\RBN$.
	Hence, if $C \trans{t} C'$ in $\ASMS$, then $\widehat{C} \trans{(a,!Read_a,a) + (p,?Read_a,q)} \widehat{C'}$ in $\RBN$.
	\item Suppose $t = p \trans{W(a)} q$ is a transition in $\ASMS$. 
	We first have two transitions $p \trans{?Ch_a} [p,a,q]$ and $[p,a,q] \trans{!Ack_a} q$.
	Further, for \textbf{every} $d \in \Sigma$, we have the transitions, $d \trans{!Ch_a} \overline{a}$
	and $\overline{a} \trans{?Ack_a} a$.
	Intuitively, the process responsible for the register requests to \emph{change} the value of the register from $d$ to $a$ by 
	broadcasting the message $Ch_a$ and moving to $\overline{a}$.
	The process at state $p$ is capable of receiving this message and moves to the state $[p,a,q]$ 
	and from there it is capable of sending the message $Ack_a$ \emph{acknowledging} the change sent by the register.
	The process at $\overline{a}$ can receive $Ack_a$ and move to $a$.
	Hence, if $C \trans{t} C'$ then $\widehat{C} \trans{(d,!Ch_a,\overline{a}) + (p,?Ch_a,[p,a,q])} C_{int} \trans{([p,a,q],!Ack_a,q) + (\overline{a},?Ack_a,a)} \widehat{C'}$.
	Figure \ref{fig:sms-to-rbn} represents the transitions needed for this simulation.\chana{ref to fig 4}
\end{itemize}

\noindent \textbf{Correctness of construction} \
Hence, if $C \trans{*} C'$ in $\ASMS$ then we have shown that $\widehat{C} \trans{*} \widehat{C'}$ in $\RBN$. Notice
that we have also shown that it is possible to go from $\widehat{C}$ to $\widehat{C'}$ where every broadcasted message is received by \emph{exactly} one other process.
Our next lemma shows that this is not an accident, and indeed any run between $\widehat{C}$ and $\widehat{C'}$ can be transformed into this form.

A run between good configurations of $\RBN$ is said to be in \emph{normal form} if whenever $Z \trans{t + t_1,\dots,t_n} Z'$ is a step
in that run, then $n = 1$. We have the following lemma.
\begin{restatable}[Normal form lemma]{lemma}{LmNormalRBN}
	Suppose there is a run from $Z$ to $Z'$ in $\RBN$ where $Z$ and $Z'$ are good configurations.
	Then there is a run from $Z$ to $Z'$ which is in normal form.
\end{restatable}

First we will see how our simulation is correct, using the normal form lemma.
Suppose $\widehat{C} \trans{*} \widehat{C'}$ in $\RBN$ for some configurations $C$ and $C'$ of $\ASMS$.
By the normal form lemma, we can assume that this run is in normal form and so 
let $\widehat{C} \trans{b^1 + r^1} C_1 \trans{b^2 + r^2} C_2 \dots C_{m-1} \trans{b^m + r^m} \widehat{C'}$.
We proceed by induction on $m$. The base case of $m = 0$ is trivial.
Suppose $m > 0$ and assume the claim holds for all numbers less than $m$.
Since $\widehat{C}$ is a good configuration, there are only two possible cases for $b^1$:

\textbf{Case 1: } Suppose $b^1 = (a,!Read_a,a)$ for some $a \in \Sigma$. 
Hence $r^1$ must be $(p,?Read_a,q)$ for some $p, q \in Q_{\ASMS}$. 
It follows that $C_1 = \widehat{Z}$ for some configuration $Z$ of $\ASMS$.
Since $C \trans{(p,R,a,q)} Z$ in $\ASMS$, by applying the induction hypothesis on the run from $\widehat{Z}$ to $\widehat{C'}$,
we are done.
%Therefore, $C_1$ is also a good configuration and so $C_1 = \widehat{Z}$ for some configuration $Z$ of $\ASMS$.
%Applying induction hypothesis we get that $C'$ is reachable from $Z$ in $\ASMS$.
%Since $C \trans{(p,R(a),q)} Z$ in $\ASMS$, it follows that $C'$ is reachable from $C$ in $\ASMS$.

\textbf{Case 2: } Suppose $b^1 = (d,!Ch_a,\overline{a})$ for some $d,a \in \Sigma$.
Hence $r^1$ must be $(p,?Ch_a,[p,a,q])$ for some $p, q \in Q_{\ASMS}$.
The only process
which can broadcast from $C_1$ is the process at $[p,a,q]$ 
and moreover it can only broadcast $Ack_a$. 
The only process which can receive $Ack_a$ from $C_1$
is the process at the state $\overline{a}$. Hence $b^2 = ([p,a,q],!Ack_a,q)$ and $r^2 = (\overline{a},?Ack_a,a)$.
Therefore, $C_2 = \widehat{Z}$ for some configuration $Z$ of $\ASMS$.
Since $C \trans{(p,W,a,q)} Z$ in $\ASMS$, by applying the induction hypothesis on the run from
$\widehat{Z}$ to $\widehat{C'}$, we are done.
%Applying induction hypothesis we get that $C'$ is reachable from $Z$ in $\ASMS$.
%Since $C \trans{(p,W(a),q)} Z$ in $\ASMS$, it follows that $C'$ is reachable from $C$ in $\ASMS$.

\begin{proof}[Proof sketch of normal form lemma]
Suppose $Z_0 := Z \trans{b^1 + r_1^1,\dots,r_{n_1}^1} Z_1 \trans{b^2 + r_1^2,\dots,r_{n_2}^2} Z_2 \dots Z_{m-1} \trans{b^m + r_1^m,\dots,r_{n_m}^m} Z_m := Z'$.
We proceed by induction on $m$. The case of $m = 0$ is trivial.
	
Suppose $m > 0$ and assume that the claim is true for all numbers less than $m$. Since $Z_0$ is a good configuration, there are only two possible choices for $b^1$.
	
\textbf{Case 1: } Suppose $b^1 = a \trans{!Read_a} a$ for some $a \in \Sigma$.
By firing $b^1$ repeatedly, we can fire $r^1_1,r^1_2,\dots,r^1_{n_1}$ ``one at a time'' and reach
$Z_1$ from $Z_0$ using a run in normal form. We can then apply the induction hypothesis on the run between $Z_1$ and $Z'$.

\textbf{Case 2: } Suppose $b^1 = d \trans{!Ch_a} \overline{a}$ for some $d,a \in \Sigma$.
Hence, $Z_1$ is a bad configuration and so $Z_1 \neq Z'$. If $n_1 = 0$, then no process in $Z_1$ can broadcast any message,
which leads to a contradiction. So, $n_1 > 0$.

For each $1 \le i \le n_1$, let $r_i^1 = (p_i,?Ch_a,[p_i,a,q_i])$.
Let $S := \sum_{i=2}^{n_1} \vec{p_i} - \sum_{i=2}^{n_1} \vec{[p_i,a,q_i]}$
and let $M := \sum_{i=2}^{n_1} \vec{[p_i,a,q_i]}$.
Notice that the only processes which can broadcast a message at the configuration $Z_1$
are the processes in the multiset $\vec{[p_1,a,q_1]} + M$. 
Hence $b^2 = (p_i[a]q_i,!Ack_a,q_i)$ for some $i$. Without loss of generality, we can assume that $i = 1$.

Notice that the only process which can receive the message $Ack_a$ at the configuration $Z_1$
is the process at the state $\overline{a}$. It then follows that either $n_2 = 0$ or $n_2 = 1$. 
Hence, we get two subcases:

\textbf{Case 2a): } Suppose $n_2 = 0$. 
Then reorder the run between $Z_0$ and $Z_2$ as follows:
$Z_0 \trans{b^1 + r_1^1} Z_1 + S \trans{b^2 + (\overline{a},?Ack_a,a)} Z_2 + S - \overline{\vec{a}} + \vec{a} 
\trans{(a,!Ch_a,\overline{a}) + r_2^1,\dots,r_{n_1}^1} Z_2$. Notice that the 
configuration $Z_2 + S - \overline{\vec{a}} + \vec{a}$ is a good configuration and 
has a run of length $m-1$ to $Z'$. Applying induction hypothesis, we are then done.

\textbf{Case 2b): } Suppose $n_2 = 1$. Hence, $r_1^2 = (\overline{a},?Ack_a,a)$ and so
$Z_2(\overline{a}) = 0$ and $Z_2(a) = 1$.
We consider two further subcases:
	\begin{itemize}
		\item Suppose there exists $\alpha > 2$ such that $Z_{\alpha}(\overline{a}) = 1$.
		Let $\alpha$ be the minimum such index. Hence, there must exist some $d' \in \Sigma$ such that $b^{\alpha}$ is $(d',!Ch_a,\overline{a})$.
		%Consider the run $Z_2 \trans{b^3 + r_1^3,\dots,r_{n_3}^3} Z_3 \trans{b^4 + r_1^4,\dots,r_{n_4}^4} \dots Z_{\alpha-1} \trans{b^{\alpha} + r_1^\alpha, \dots,r_{n_\alpha}^\alpha} Z_\alpha$.

		Suppose no $b \in \{b^3,\dots,b^{\alpha-1}\}$ is labelled by $!Ack_a$. 
		Intuitively, we can then show that none of the processes in any of the states
		in the multiset $M$ ever make a step between $Z_2$ and $Z_\alpha$. 
		Hence, we can ``postpone'' firing the transitions $r^1_2,\dots,r^1_{n_1}$ and get
		$Z_0 \trans{b^1 + r_1^1} Z_1 + S \trans{b^2 + r_1^2} Z_2 + S \act{*} Z_{\alpha-1}+S \act{b^\alpha + r^1_2,\dots,r^1_{n_1},r_1^{\alpha},\dots,r_{n_\alpha}^\alpha} Z_\alpha$. The configuration $Z_2 + S$ is a good configuration and has a run to $Z'$ of length $m-2$ and so we can apply the induction hypothesis.
		
		Suppose some $b \in \{b^3,\dots,b^{\alpha-1}\}$ is labelled by $!Ack_a$. Let $b = b^i$ be the first such transition.
		By definition of $\alpha$ and by construction of the protocol, we can show that $b^i$ must be $([p_j,a,q_j],!Ack_a,q_j)$ for some $2 \le j \le n_1$ (without loss of
		generality we can assume $j = 2$) and we can also show that no process at the $i^{th}$ step receives this message, i.e. $n_i = 0$.
		Hence, we can ``prepone'' firing the transition $b^i$ and get $Z_0 \trans{b^1 + r_1^1} Z_1 + S \trans{b^2 + r_1^2} Z_2 + S
		\trans{(a,!Ch_a,\overline{a}) + r_2^1,\dots,r_{n_1}^1} Z_2 - \vec{a} + \overline{\vec{a}} 
		\trans{(b^i + (\overline{a},?Ack_a,a)} Z_2 - \vec{[p_2,a,q_2]} + \vec{p_2} \act{*} 
		Z_{i-2} - \vec{[p_2,a,q_2]} + \vec{p_2} \trans{b^{i-1} + r_1^{i-1} \dots r_{n_{i-1}}^{i-1}} Z_i$.
		Notice that $Z_2 + S$ is a good configuration and has a run to $Z'$ of length $<m$ and so we can apply the induction hypothesis.
		
		\item Suppose there does not exist $\alpha > 2$ such that $Z_{\alpha}(\overline{a}) = 1$.
		
		Suppose no $b \in \{b^3,\dots,b^{m}\}$ is labelled by $!Ack_a$. We can once again show that none of the processes in the multiset $M$ 
		ever make a step between $Z_2$ and $Z_m$. Since $Z_m$ is a good configuration, it must then be the case that $n_1 = 1$, which means that
		$Z_2$ is a good configuration and the run between $Z_0$ and $Z_2$ is already in normal form. Because of the induction hypothesis, we are done.
		 
		Suppose some $b \in \{b^3,\dots,b^{m}\}$ is labelled by $!Ack_a$. Let $b = b^i$ be the first such transition.
		In this case, we can do a similar rearrangement like the corresponding previous case by ``preponing'' $b^i$ 
		and then conclude by applying the induction hypothesis.
	\end{itemize}
\end{proof}

\noindent \textbf{The reduction} \
Now, suppose we are given an ASMS $\ASMS$ and two cubes $\cube_1 = (L_1,U_1,d)$ and $\cube'_1 = (L'_1,U'_1,d')$. 
We construct the protocol $\RBN$ as we have described in this section. 
Then we construct two cubes $\cube_2 = (L_2,U_2)$ and $\cube'_2 = (L_2',U_2')$ of $\RBN$ as follows: 
$L_2(q), U_2(q), L_2'(q)$ and $U_2'(q)$ are all respectively equal to $L_1(q), U_1(q), L_1'(q)$ and $U_1'(q)$ if $q \in Q_\ASMS$, 
$L_2(d) = U_2(d) = L_2'(d') = U_2'(d') = 1$ and otherwise $L_2(q) = U_2(q) = L_2'(q) = U_2'(q) = 0$.
It is easy to see that a configuration $C \in \cube_1$ (resp. $\cube'_1$) iff its corresponding 
configuration $\hat{C} \in \cube_2$ (resp. $\cube'_2$). Hence, by our simulation it follows that
$\cube_1$ can reach $\cube'_1$ in $\ASMS$ iff $\cube_2$ can reach $\cube'_2$ 
in $\RBN$.\\

\noindent \textbf{Another reduction} \
While this construction proves the desired result, we need a slightly different construction for the purposes of the next section which we now describe. Given an ASMS $\ASMS$ and 
two cubes $\cube_1 = (L_1,U_1,d)$ and $\cube'_1 = (L_1',U_1',d')$, once again construct the RBN $\RBN$ described in this section
and construct two cubes $\cube_3 = (L_3,U_3)$ and $\cube'_3 = (L_3',U_3')$ of $\RBN$ as follows:
The cube $\cube'_3$ is the same as $\cube'_2$ described before. The cube $\cube_3$ is also
exactly the same as $\cube_2$, except for the constraints $L_2(d) = U_2(d) = 1$ which
are replaced by $L_3(d) = 0, U_3(d) = \infty$.

Since $\cube_2 \subseteq \cube_3$, it follows from the previous reduction that if $\cube_1$ can reach $\cube_1'$,
then $\cube_3$ can reach $\cube_3'$.
For the other direction, notice that, by construction of the protocol $\RBN$,
\begin{equation}\label{eq:invariant}
	\text{If } C \act{} C' \text{ is a step in } \RBN,
	\text{ then } \sum_{q \in \Sigma \cup \overline{\Sigma}} C(q) = \sum_{q \in \Sigma \cup \overline{\Sigma}} C'(q)
\end{equation}
Using this equation and the fact that any configuration in $\cube_3'$ is a good configuration,
it is then clear that if $C_3 \in \cube_3$ such that $C_3 \act{*} C_3'$ with $C_3' \in \cube_3'$,
then $C_3$ must also be a good configuration. 
Hence, we can then conclude 
%Now, it is easy to see that if $\cube_1$ can reach $\cube_1'$ in $\ASMS$ then $\cube_3$ can reach
%$\cube'_3$ in $\RBN$. Now suppose there is a configuration $C_3 \in \cube_3$ which reaches
%some configuration $C'_3 \in \cube_3'$. By construction of $\cube_3'$ and $\cube_1'$, 
%$C'_3$ must be $\hat{C'_1}$ for some $C'_1 \in \cube'_1$. Equation~\ref{eq:invariant} along
%with the fact $C'_3$ is a good configuration must then imply that $C_3$ is also a good configuration
%and is hence equal to $\hat{C_1}$ for some $C_1$. 
%The only reason why $C_3$ cannot be 
%a good configuration is if $C_3(d) \neq 1$. Notice that $C_3(d) > 0$, as otherwise no run from 
%$C_3$ is possible. If $C_3(d) > 1$, then by applying equation~\ref{eq:invariant} inductively, we get that
%$\sum_{q \in \Sigma \cup \overline{\Sigma}} C_3'(q) > 1$, which is a contradiction since $C_3' \in \cube'_3$.
%Hence, $C_3(d) = 1$ and so $C_3$ must be $\hat{C_1}$ for some $C_1 \in \cube_1$.
%It then follows that since $C_3 \act{*} C_3'$, we have $C_1 \act{*} C_1'$. Hence, it follows
that $\cube_3$ can reach $\cube_3'$ iff $\cube_1$ can reach $\cube_1'$.

%%%%%%%%%%%%%%%%%%%%%%%%%%%%%l
\section{Transferring Existing Results}\label{sec:consequences}
In the previous section, we have shown that RBN and ASMS are polynomial-time equivalent with respect \chana{updated} to the cube-reachability problem. 
Though the precise complexity of this 
problem has not been established for either one of these models, our result shows that it is sufficient to characterize the complexity of cube-reachability
for one of these models. 
Moreover, there exist results for subclasses \chana{updated}
of the cube-reachability problem for both RBN and ASMS. In this section, we use the reductions constructed in the previous section to transfer these results from RBN to ASMS and vice versa. 
%\bala{Added new paragraph. We had to talk about the complexity status of the cube-reachability problem at some point.} \chana{looks good}

\subsection{Unbounded initial cube reachability}
\label{subsec:unbounded-cube}

We consider the following problem for RBN, which we call the \emph{unbounded initial cube reachability} problem: We are given an RBN $\RBN$ and two cubes $\cube = (L,U), \cube' = (L',U')$ with the special property that $L(q) = 0$ and $U(q) \in \{0,\infty\}$ for every state $q$ and 
we would like to check if $\cube$ can reach $\cube'$. Notice that there is no restriction on the cube $\cube'$. 
We will call such a pair $(\cube,\cube')$ as an \emph{unbounded initial cube} pair.
This problem was proved to be \textsf{PSPACE}-complete for RBN in (\cite{FSTTCS12}, Theorem 5.5).
(In~\cite{FSTTCS12}, this result is only stated for cubes with constants encoded in unary, but the proof
can be modified easily to also give the same upper bound when the constants are encoded in binary).

In a similar way, it is possible to define the corresponding problem for ASMS. 
%We are given an ASMS $\ASMS$ and two cubes $\cube = (L,U,d), \cube' = (L',U',d')$ such that $L(q) = 0$ and $U(q) \in \{0,\infty\}$ for every state $q$ of $\ASMS$ and we would like to decide if $\cube$ can reach $\cube'$. 
Notice that if $(\cube,\cube')$ is an unbounded initial cube pair for an ASMS $\ASMS$, then
the second reduction in Section~\ref{subsec:RBN-ASMS}  produces an RBN $\RBN$ along with an unbounded initial cube pair as well. This shows that the corresponding problem for ASMS can be solved in \textsf{PSPACE}. 

Further, notice that given an RBN $\RBN$ and an unbounded initial cube pair for $\RBN$, our reduction in Section~\ref{subsec:ASMS-RBN} produces an ASMS $\ASMS$ with an unbounded initial cube pair as well. This shows that the unbounded initial cube reachability problem for ASMS is \textsf{PSPACE}-hard.

\begin{theorem}
	The unbounded initial cube reachability problem for ASMS is \textsf{PSPACE}-complete.
\end{theorem}

\subsection{Leader protocols}\label{subsec:leader}
The notion of an ASMS equipped with a \emph{leader} has been studied in~\cite{JACM16, FineGrained}. Formally, an ASMS-leader protocol is a pair of ASMS protocols $\ASMS_C = (Q_C,\Sigma,\delta_C), 
\ASMS_D = (Q_D,\Sigma,\delta_D)$, 
where $\ASMS_C$ is called the \emph{contributor protocol} and $\ASMS_D$ is called the \emph{leader protocol}.
Intuitively, there is exactly one process which executes $\ASMS_D$ (the leader) and all the other processes execute $\ASMS_C$ (contributors).
This is formalized as follows: A configuration of such a system is defined to be a triple $(q,M,a)$ where $q \in Q_D$, $M$ is a multiset on $Q_C$
and $a \in \Sigma$. A step between $C = (q,M,a)$ and $C' = (q',M',a')$ exists if one of the following is true:
\begin{itemize}
	\item There exists $(q,\oper,a',q') \in \delta_D$ such that $M' = M$ and either $\oper = R$ and $a = a'$, or $\oper = W$.
	\item There exists $(p,\oper,a',p') \in \delta_C$ such that $q = q'$, $M(p) \ge 1$, $M' = M - \vec{p} + \vec{p'}$, and either $\oper = R$ and $a = a'$, or $\oper = W$.
\end{itemize}

We can then define the notion of a run for an ASMS-leader protocol in the usual way. The \emph{ASMS-leader reachability} problem is to decide, given an ASMS-leader protocol $(\ASMS_C,\ASMS_D)$, two leader states $q_D^I, q_D^f$, a contributor state $q_C^I$ and two data values $a,a' \in \Sigma$ whether there exists a $k \ge 1$ such that the configuration $(q_D^I,\multiset{k \cdot q_C^I},a)$ can reach a configuration $C' = (q_D^f,M',a')$ for some $M'$.

We now define a special case of cube-reachability in ASMS and notice that this special case is exactly equivalent to ASMS-leader reachability.
An \emph{ASMS-leader cube} is a pair $(\ASMS,\cube,\cube')$ of the following form: The protocol $\ASMS = (Q,\Sigma,\delta)$ is such that there exists a partition of the states and transition relation as $Q = Q_C \cup Q_D, \delta = \delta_C \cup \delta_D$ 
and $\cube = (L,U,a), \cube' = (L',U',a')$ satisfy: There exists exactly two states $q_D^I, q_D^f \in Q_D$ such that $L(q_D^I) = U(q_D^I) = 1, L'(q_D^f) = U'(q_D^f) = 1$ %\chana{there are two L which mean different things, can we find an alternative notation for one of them?} \bala{Good call. Has been changed now.}
and for every other state $q \in Q_D$, $L(q) = L'(q) = U(q) = U'(q) = 0$ and there exists exactly one state $q_C^I \in Q_C$ such that $L(q_C^I) = L'(q_C^I) = 0, U(q_C^I) = U'(q_C^I) = \infty$ and for every other state $q \in Q_C$, $L(q) = U(q) = L'(q) = 0, U'(q) = \infty$. 
Notice that Example~\ref{ex:asms} is an example of an ASMS-leader cube.

It is easy to see that the ASMS-leader reachability problem is equivalent to the cube-reachability problem for ASMS-leader cubes.
%The following is known about the ASMS-leader reachability problem~\cite{JACM16}:
The following result has been shown for ASMS.
\begin{theorem}[\cite{JACM16}]
	The ASMS-leader reachability problem is in \textsf{NP}.
\end{theorem}

Now, we can define RBN-leader protocols and RBN-leader cubes in exactly the same way as was done for ASMS. Further, notice that the reduction given in Section~\ref{subsec:ASMS-RBN} has the following special property:
If we are given an RBN-leader cube $(\RBN,\cube_1,\cube'_1)$, then the reduction produces an ASMS-leader cube $(\ASMS,\cube_2,\cube_2')$. This proves that the RBN-leader cube reachability problem (and hence the RBN-leader
reachability problem) is in \textsf{NP}. 

Notice that the reduction given in Section~\ref{subsec:RBN-ASMS} does not output a RBN-leader cube when it is given an ASMS-leader cube as input. Hence, we do not immediately get \textsf{NP}-hardness of the 
RBN-leader reachability problem. Nevertheless, by a reduction from 3-SAT similar to that of the one given in Theorem 10 of~\cite{FineGrained}, we can prove \textsf{NP}-hardness of RBN-leader reachability. We then get
\begin{restatable}{theorem}{ThmLeaderRBN}
	The RBN-leader reachability problem is \textsf{NP}-complete.
\end{restatable}

%\chana{should there be a  transfer to IO-leader nets?}

\subsection{Almost-sure coverability}

We now consider the notion of \emph{almost-sure coverability} for ASMS. Let $\ASMS = (Q,\Sigma,\delta)$ be an ASMS with two distinguished states $q_I, q_f$ and 
a distinguished letter $d \in \Sigma$. Let $\uparrow q_f$ denote the set of all configurations $C$ such that $C(q_f) \ge 1$.
For any $k \ge 1$, we say that the configuration $(\multiset{k \cdot q_I},d)$ \emph{almost-surely covers} $q_f$ iff $\post^*((\multiset{k \cdot q_I},d)) \subseteq \pre^*(\uparrow q_f)$. 
The reason behind calling this the almost-sure coverability relation is that the definition given
here is equivalent to covering the state $q_f$ from $(\multiset{k \cdot q_I},d)$ with probability 1
under a probabilistic scheduler which picks processes uniformly at random at each step.
%The intuition behind this definition is the following : Based on the transition relation between configurations, it is possible to define a Markov chain where the 
%probability of moving from  $C$ to  $C'$ is 0 if there is no edge between $C$ and $C'$ and $1/|\post(C)|$ %otherwise. 
%Within this Markov chain, it is then easy to prove that the probability that a path starting at $(\multiset{k \cdot q_I},d)$ reaches a configuration in $\uparrow q_f$ is 1
%iff $\post^*((\multiset{k \cdot q_I}, d)) \subseteq \pre^*(\uparrow q_f)$. Since $\uparrow q_f$ is exactly the set of configurations which cover $q_f$, this is said to be the almost-sure coverability relation.

The number $k$ is called a \emph{cut-off} if one of the following is true: 1) Either for all $h \ge k$, the configuration $(\multiset{h \cdot q_I},d)$ almost-surely covers $q_f$. In this case, $k$ is a positive cut-off. Or, 2) for all $h \ge k$, the configuration $(\multiset{h \cdot q_I},d)$ does not almost-surely cover $q_f$. In this case, $k$ is a negative cut-off.
Note that from the definition alone, it is not clear that a cut-off must exist for every ASMS. The following result is known.
\begin{theorem}[Theorem 3 of~\cite{ICALPPatricia}]~\label{thm:ICALPPatricia}
	Given an ASMS with two states $q_I,q_f$ and a letter $d$, a cut-off always exists. Whether the cut-off is positive or negative can be decided in \textsf{EXPSPACE}.
\end{theorem}

We can now translate this result to RBNs. Given an RBN $\RBN$ and two states $q_I, q_f$, we first set $\uparrow q_f := \{C : C(q_f) \ge 1 \}$.
Then for any $k \ge 1$, we say that $\multiset{k \cdot q_I}$ almost-surely covers $q_f$ iff $\post^*(\multiset{k \cdot q_I}) \subseteq \pre^*(\uparrow q_f)$. 
We can then define positive and negative cut-offs in a similar manner. %We can now use the simulation of an RBN by an ASMS given in Section~\ref{subsec:ASMS-RBN} to show that a cut-off for an RBN always exists and deciding whether or not the cut-off is positive can be done in \textsf{EXPSPACE}.
%Let $\RBN$ be a RBN with two states $q_I,q_f$ and 
Now for the RBN $\RBN$, let $\ASMS$ be the ASMS protocol that we construct in our reduction given in Section~\ref{subsec:ASMS-RBN}. Using the construction of $\ASMS$, we can then easily show that
\begin{quote}
    for any $k \ge 1$, $\post^{*}_{\RBN}(\multiset{k \cdot q_I}) \subseteq \pre^{*}_{\RBN}(\uparrow q_f)$ iff $\post^{*}_{\ASMS}(\multiset{k \cdot q_I},\#) \subseteq \pre^{*}_{\ASMS}(\uparrow q_f)$.    
\end{quote}
This then directly implies that
\begin{restatable}{theorem}{ThmASCoverabilityRBN}
	Given an RBN with two states $q_I,q_f$, a cut-off always exists. Whether the cut-off is positive or negative can be decided in \textsf{EXPSPACE}.
\end{restatable}
%\chana{should there be a  transfer to IO nets?}

%%%%%%%%%%%%%%%%%%%%%%%%%%%%%
\section{A Third Model}\label{sec:IO}
We have shown that RBN and ASMS are cube-reachability equivalent and using this we have transferred some results between these two models.
In this section, we will introduce a third model called \emph{Immediate Observation} (IO) nets and show that cube-reachability for IO nets can be reduced to cube-reachability for RBN.
Further, we show that a stronger notion of reduction -- which is satisfied by all the reductions given in this paper -- cannot exist from RBN to IO nets.

\subsection{Immediate Observation Nets}
%!TEX root = main.tex
Immediate observation nets, or \emph{IO nets}, were introduced in~\cite{conf/apn/EsparzaRW19}.
They are a subclass of Petri nets with applications in population protocols and chemical reaction networks. 
An IO net is a Petri net with transitions of a certain shape:
Informally, %\bala{Tokens and places have not been defined. Should we use processes and states?} \chana{changed}
a process (or token) in a state (or place) $p$ \emph{observes} the presence of a process in $q$ and moves to state $p'$, for some states $p,q,p'$ not necessarily distinct.
Because of this, IO nets
%can be viewed as a finite state transition system instead of as a Petri net. 
can be described in a simpler manner that does not use the full Petri net formalism. 
%\bala{Need to change this to something else other than finite state transition system} \chana{tried something}
We will present them this way here, 
to highlight the similarity to the other models and 
to simplify notation.

\begin{definition}
\label{def:io}
An \emph{immediate observation net} is a tuple $\net = (Q,\delta)$ 
where $Q$ is a finite set of states 
and $\delta \subseteq Q \times Q \times Q$ is the transition relation.
\end{definition}

If $(p,q,p') \in \delta$, then we sometimes denote it by $p \act{q} p'$.
A \emph{configuration} $C$ of an IO net $\net$ is a multiset over $Q$.
It intuitively counts the number of processes in each state. 
There is a \emph{step} from a configuration $C$ to a configuration $C'$
if there exists $t=p \trans{q} p' \in \delta$,
%if $C(p) \ge 1, C(q) \ge 1$  
such that $C \ge \multiset{p,q}$
and $C' = C - \vec{p} + \vec{p'}$. 
We denote by $C \trans{t} C'$ such a step, and by $\trans{*}$ the reflexive transitive closure of the step relation. 
%There is a \emph{run} $\sigma$ from a configuration $C$ to a configuration $C'$
%if there exists configurations $C_1, C_2 \ldots, C_n$ and transitions $t_1, t_2 \ldots, t_{n-1} \in \delta$
%such that $C=C_1 \trans{t_1} \ldots \trans{t_{n-1}} C_n$.
%We denote this by $C \trans{\sigma} C'$.
We can then define runs of an IO net in the usual way.

%\begin{example}
%\label{ex:io}
%An example (with figure) + example of a run. 
%\chana{Do we retain the Petri net notation for the IO figures?}
%\end{example}
\subsection{RBN Simulate IO Nets}
%!TEX root = main.tex
\textbf{Construction} \ Let $\net = (Q,\delta)$ be an IO net. 
We construct an RBN that simulates $\net$
in which processes send messages signaling their current state.
Let $\RBN = (Q', \Sigma',\delta')$ be the following RBN: 
The set of states $Q'$ and the alphabet $\Sigma'$ are both  equal to $Q$.
The transition relation $\delta'$ is such that 
for every $q \in Q$ 
there is a transition $q \trans{!q} q$ in $\delta'$, and
for every $p \trans{q} p' \in \delta$ 
there is a transition $p \trans{?q} p'$ in $\delta'$.

%\begin{figure}
%\input{fig-io-to-rbn}
%\caption{Simulation in $\RBN$ of transition $p\trans{q}p'$ of $\net$.}
%\label{fig:io-to-rbn}
%\end{figure}

\noindent \textbf{Correctness of construction} \ There is a natural bijection between configurations of $\RBN$ and configurations of $\ASMS$. 
If $C$ is a configuration of $\net$, we will abuse notation and denote the corresponding  configuration of $\RBN$ also as $C$.
We now show that $C' \in \poststar_{\net}(C)$ iff $C' \in \poststar_{\RBN} (C)$ for any configurations $C$ and $C'$ of $\net$.
Indeed, if $C$ reaches $C'$ by one step $p\trans{q} p'$ in $\net$, 
then $C \trans{t+t_1} C'$ with $t=q \trans{!q} q$ and $t_1= p \trans{?q} p'$ in $\RBN$. 
%\bala{Can you change this so that it is consistent with the arrow notation?}\chana{done}
%Thus, if $C \trans{*} C'$ for some configurations $C,C'$ in $\net$,
%then $C \trans{*} C'$ in $\RBN$. 
Conversely,  let $C \trans{t+t_1,\ldots, t_k} C'$ be a step  in $\RBN$ with
$t= q \trans{!q} q$ and $t_i=p_i \trans{?q} p'_i$ 
for some $k \ge 0$.
The step must be of this form because the only broadcast transitions of $\RBN$
are of the form $q \trans{!q} q$.
Then $C$ reaches $C'$ by the sequence of transitions
$ (p_1 \trans{q} p'_1 ), (p_2 \trans{q} p'_2), \ldots ,( p_k \trans{q} p'_k)$ in $\net$. 

\noindent \textbf{The reduction}  \ With this construction, RBN can simulate IO nets as follows: 
Let $\net$ be an IO net and let $\cube_1 = (L_1,U_1),\cube_1' = (L_1',U_1')$ be two cubes of $\net$. 
Construct the RBN $\RBN$ as described above,
and let $\cube_2 = \cube_1$ and $\cube_2' = \cube_1'$.
By our construction, $\cube_1$ can reach $\cube_1'$ iff $\cube_2$ can reach $\cube_2'$.

\paragraph*{Consequences.}
In~\cite{FSTTCS12}, 
two further restrictions of the unbounded initial cube reachability problem
(presented in Section \ref{subsec:unbounded-cube})
are considered. 
The first restriction, dubbed $CRP[\ge1]$
(where CRP stands for cardinality reachability problem), 
considers only unbounded initial cube pairs $\cube, \cube'$
in which $\cube'=(L',U')$ is such that 
$L'(q) \in \set{0,1}$ and $U'(q) =\infty$ for all $q$.
The second restriction, dubbed $CRP[\ge1,=0]$, 
considers only unbounded initial cube pairs $\cube, \cube'$
in which $\cube'=(L',U')$ is such that 
$L'(q) \in \set{0,1}$ and $U'(q) \in \set{0,\infty}$ for all $q$.
For RBN, the problems $CRP[\ge 1]$ and $CRP[\ge 1,=0]$ are 
shown to be in PTIME and in NP (Theorem 3.3 and 4.3 of~\cite{FSTTCS12}), respectively.
By the construction given above, it is then immediately clear that
%The same results are obtained for IO nets
%by using the above simulation.
%The proofs are immediate because our simulation 
%is constructible in polynomial time and verifies that 
%$C' \in \poststar_{\net}(C)$ iff $C' \in \poststar_{\RBN} (C)$
%for any configurations $C,C'$ over $Q$.

\begin{theorem}
\label{thm:io-p}
For IO nets, $CRP[\ge 1]$ and $CRP[\ge 1,=0]$ are in PTIME and in NP respectively.
\end{theorem}

\paragraph*{Strong simulation.}
Consider the following alternative definition of simulation between models $A$ and $B$ (where a model is to be understood as either an RBN, ASMS or IO net): Given an instance $I$ of model $A$ with states $Q_I$,
there exists an instance $J$ of model $B$ with states $Q_J$ such that $J$ is polynomial in the
size of $I$
%Given models $A$ and $B$ with finite set of states $Q_A$ and $Q_B$, 
%model $B$ \emph{strongly simulates} model $A$ 
%if $B$ is of size polynomial in the size of $A$ 
with $Q_I  \subseteq Q_J$, and
%there exists a vector $h \in \N^{Q_B}$, where $Q_B$ is the finite state set of $B$,
%such that for any configurations $C,C'$ of $A$,
%$C' \in \poststar(C)$ if and only if $C'+h \in \poststar(C+h)$.
%
%This definition entails a bijection 
%from the configurations of $A$ to a 
%set of ``good configurations" $\mathcal{B}$ in $B$.
%The bijection verifies that a cube of $A$ is 
%mapped to a cube of $B$, and 
%that a cube of $B$ 
%restricted to configurations of $\mathcal{B}$
%is mapped to a cube of $A$.
%The simulation constructions from IO nets to RBN and from ASMS to RBN 
%verify this strong definition of simulation.
%\chana{it's weird that the third and last simulation is not there...}
there exists a multiset $h$ over $Q_J\setminus Q_I$ of polynomial size
such that 
$C' \in \poststar(C)$ if and only if $C'\cdot h \in \poststar(C\cdot h)$
for any configurations $C,C'$ of $I$.
Notice that strong simulation is a transitive relation. 
%\chana{checked}
%\bala{Do you need constructibility of $B$ in polynomial time?} \chana{I think not, I only use strong sim in Thm 14}
%
%Intuitively, this definition entails a bijection $b$ from configurations of $A$ to 
%a subset $\image$ of ``good" configurations of $B$.
%The bijection verifies that a cube of $A$ is
%mapped to a cube of $B$, and 
%that a cube of $B$ 
%restricted to configurations of $\mathcal{B}$
%is mapped to a cube of $A$.
%Intuitively, the image of a cube of $A$ 
%is its ``concatenation" with the cube on $Q_B\setminus Q_A$ of lower and upper bound equal to $h$.
%A cube $\cube$ of $B$ restricted to $\image$
%is equal to the cube $\cube \cap \mathcal{H}$ where
%$\mathcal{H}$ is the cube 
%of lower bound $0$ and upper bound $\infty$ on $Q_A$, and upper and lower bounds equal to $h$ on $Q_B \setminus Q_A$. 
%The reverse image of this cube is the cube of $A$ in which we ``forget" the information of $Q_B \setminus Q_A$. 
%\chana{L,U for the empty cube not defined yet} 
%\bala{On second thought, I think we can remove this paragraph or shorten it even more}
%
The simulation constructions 
of this paper
%from IO nets to RBN and from ASMS to RBN 
verify this strong definition of simulation.

\begin{restatable}{theorem}{StrongSimAll}
\label{thm:strong-sim}
RBN and ASMS strongly simulate each other. Further, IO nets are strongly simulated by RBN (and hence by ASMS as well).
%Given an RBN $\RBN$, there exists an ASMS $\ASMS$ such that $\ASMS$ strongly simulates $\RBN$. 
%Given an ASMS $\ASMS$, there exists an RBN $\RBN$ such that $\RBN$ strongly simulates $\ASMS$.
\end{restatable} 
 
We show that this is not the case for IO nets: they cannot strongly simulate RBN (nor ASMS).

\subsection{IO Does not Strongly Simulate RBN}
%!TEX root = main.tex
Assuming that IO nets can strongly simulate RBN,
we will derive a contradiction.
Under this assumption, 
we will first transfer results on the closure of cubes from IO nets to RBN, 
then exhibit a particular RBN which contradicts these results.
We start by recalling definitions and properties relating to cubes.

\paragraph*{Counting sets and norms.}
We consider cubes over a finite set $Q$.
A
%formal
finite union of cubes $\bigcup_{i=1}^m (L_i,U_i)$ is called a \emph{counting constraint}
and the set of configurations $\bigcup_{i=1}^m \cube_i$ it describes is called a \emph{counting set}.
%formal enough?
We write $\sem{\cC}$ for the counting set described by the counting constraint $\cC$.
Notice that two different counting constraints may describe the same counting set.
For example, let $Q=\set{q}$ and let $(L,U)=(1,3)$, $(L',U')=(2,4)$, $(L'',U'')=(1,4)$. 
The counting constraints $(L,U)\cup(L',U')$ and $(L'',U'')$ define the same counting set.
%the following sentence is directly from PN2019, sorry plagiarism gods
It is easy to show (see also Proposition 2 of \cite{EsparzaGMW18})
that counting constraints and counting sets are closed under Boolean operations.

Let $\cube=(L,U)$ be a cube.
Let $\lnorm{\cube}$ be the the sum of the components of $L$.
Let $\unorm{\cube}$ be the sum of the finite components of $U$ if there are any, and $0$ otherwise.
We call \emph{norm} of $\cube$ the maximum of $\lnorm{\cube}$ and $\unorm{\cube}$, denoted by $\norm{\cube}$.
We define the norm of a counting constraint $\cC= \bigcup_{i=1}^m \cube_i$ as
%in the following way 
%$$
$\norm{\cC} \defeq \displaystyle \max_{i\in [1,m]} \{ \norm{\cube_i} \}$.
%$$
%the next few lines are almost those from the DC submission, sorry plagiarism gods
The norm of a counting set $\cSet$ is the smallest norm of a counting constraint representing $\cSet$, that is, 
$\norm{\cSet} \defeq \displaystyle \min_{\cSet = \sem{\cC}} \{ \norm{\cC} \}$.
Proposition 5 of~\cite{EsparzaGMW18} entails the following results for the norms of the union, intersection and complement.
\begin{prop}%
\label{prop:oponconf}
Let $\cSet_1, \cSet_2$ be counting sets.
The norms of the union, intersection and complement satisfy:
$\norm{\cSet_1 \cup \cSet_2} \leq \max \{\norm{\cSet_1}, \norm{\cSet_2} \}$,
$\norm{\cSet_1 \cap \cSet_2} \leq \norm{\cSet_1} + \norm{\cSet_2}$
and
$\norm{\N^n \setminus \cSet_1} \leq \norm{\cSet_1} + \norm{\cSet_2}$.
%\begin{itemize}
%\item The norm of  counting set $\cSet = \cSet_1 \cup \cSet_2$ is such that
%$\norm{\cSet} \leq \max \{\norm{\cSet_1}, \norm{\cSet_2} \}$.
%
%\item  The norm of  counting set $\cSet = \cSet_1 \cap \cSet_2$ is such that
%$\norm{\cSet} \leq \norm{\cSet_1} + \norm{\cSet_2}$.
%
%\item The norm of  counting set $\cSet = \N^n \setminus \cSet_1$ is such that
%$\norm{\cSet} \leq n\norm{\cSet_1} + n$.
%
%\end{itemize}
\end{prop}
%\chana{shortened}

The following result for IO nets is deduced directly from Theorem 6 in \cite{EsparzaRW19}.
It states that the forward and backward reachability set of a counting set is still a counting set, and bounds its norms polynomially.
This result, transferred to RBN under the assumption of a strong simulation,
will amount to a contradiction.

%from DC one-way-comm.. thm 5.4 BUT rewritten to use that we defined norms on counting sets
\begin{theorem}
\label{thm:ccreach-io}
Let $\net=(Q,\delta)$ be an IO net, and let $\cSet$ be a counting set of $\net$.
Then $\poststar(\cSet)$ is also a counting set and
$
\norm{\poststar(\cSet)} \leq  \norm{\cSet} + |Q|^3.
$
The same holds for $\prestar(\cSet)$.
\end{theorem}

Assuming that IO nets strongly simulate RBN, we can transfer the result of 
Theorem \ref{thm:ccreach-io} to RBN.
 
\begin{restatable}{theorem}{ThmCCReachRBN}
\label{thm:ccreach-rbn}
Assume that IO nets can strongly simulate RBN.
There exists a constant $k$ such that 
for any RBN $\RBN=(Q,\Sigma,\delta)$, 
for any counting set $\cSet$ of $\RBN$,
 $\poststar(\cSet)$ is also a counting set and 
$
\norm{\poststar(\cSet)}  \in O( \norm{\cSet}  + |Q|)^{k}.
$
The same holds for $\prestar(\cSet)$.
\end{restatable}
\begin{proof}[Proof Sketch.]
\newcommand{\image}{\mathcal{G}}
\newcommand{\complicated}{\mathcal{M}}

It suffices to show the result for $\cSet$ a cube, 
since for a counting set $\cup_{i=1}^n \cube_i$, 
we have $\poststar(\cup_i \cube_i)=\cup_i  \poststar(\cube_i)$.
Fix an RBN $\RBN=(Q,\Sigma,\delta)$ and a cube $\cube$ over $Q$.
Let $\net=(Q_\net, \delta_\net)$ be the IO net of the strong simulation whose existence we assume. 
The definition of strong simulation 
entails the existence of a bijection $b$ from configurations of $\RBN$ to 
a subset $\image$ of ``good" configurations of  $\net$.
The bijection verifies that a cube of $\RBN$ is
mapped to a cube of $\net$, and 
that a cube of $\net$ 
restricted to configurations of $\image$
is mapped to a cube of $\RBN$.

Since $b$ preserves cubes, $b(\cube)$ is a cube.
By Theorem \ref{thm:ccreach-io}, $\poststar(b(\cube))$ is a counting set,
and thus there exist cubes $\cube_1, \ldots, \cube_n$ of $\net$ 
such that $\poststar(b(\cube))=\cup_{i=1}^n \cube_i$.
Let $\complicated$ be the set $\cup_{i=1}^n b^{-1}(\cube_i|_\image)$ of $\RBN$.
We show that $\poststar(\cube)=\complicated$.
Since the $b^{-1}(\cube_i|_\image)$ are cubes by strong simulation,
$\poststar(\cube)$ is a counting set as a union of cubes.
The size of $\poststar(b(\cube))$ is polynomial in $\cube$ 
and $\RBN$ by Theorem  \ref{thm:ccreach-io},
and thus the size of $\poststar(\cube)$ is too.
\end{proof}

%\begin{ igure}[h!]
%\centering
%\includegraphics[scale=0.75]{fig-asms-counter.png}
%\caption{Graphic design is my passion.}
%\label{fig:fig-rbn-counter}
%\end{figure}

\paragraph*{Deriving the contradiction.}
We now exhibit a contradiction to the result of Theorem 
\ref{thm:ccreach-rbn}, thus proving that IO nets do not strongly simulate RBN.
%The RBN presented  is inspired by the ASMS described in Section 
%5.1 of \cite{icalp16}.
Recall the RBN represented in Figure \ref{fig:rbn}.
We can generalize it to a family of RBN $\RBN_n=(Q,\Sigma,\delta)$, 
parameterized by $n\ge 1$,
with set of states $\set{tok, sent} \cup \set{a_i,b_i,c_i | 1 \le i \le n}$.
%We define two cubes.
Let $\cube_0$ be the cube in which there are arbitrarily many agents in $tok$,
exactly one agent in each $a_i$ and $0$ agents in the other states.
Let $\cube_f$ be the cube in which there is a least one agent in $c_n$
and an arbitrary number elsewhere.
We claim that if we start from a configuration of $\cube_0$,
we can only reach $\cube_f$ if we initially have $2^n$ or more agents in $tok$.
Indeed we can show by  induction on $i \in \set{1,\ldots,n}$
that $1$ must be broadcasted $2^i$ times to reach $c_i$,
and thus that  $2^i$ agents are needed in $tok$ 
initially to reach $c_i$. %\chana{enough?}
By Proposition \ref{prop:oponconf} and Theorem \ref{thm:ccreach-rbn},
the set $S := \poststar(\cube_0) \cap \cube_f$ is a counting set of size at most polynomial in $|Q|, \norm{\cube_0}$ and  $\norm{\cube_f}$.
The cubes $\norm{\cube_0}$ and  $\norm{\cube_f}$ have norms 
$n$ and $1$  respectively, so  $S$ is of norm polynomial in $n$.
Thus if it is non-empty 
it must contain a configuration of size at most polynomial in $n$:
simply take the configuration equal to the lower bounds $L$ of one of the cubes whose union is the counting set $\poststar(\cube_0) \cap \cube_f$.
%Our claim is contradicted.
This contradicts the fact that $2^n$ agents are needed to reach $\cube_f$.
%We prove our claim: that we can reach $\cube_f$ from $\cube_0$
%if and only if we initially have $2^n$ or more agents in $tok$.
%Clearly, if we start in a configuratioin of $\cube_0$ with $2^n$ agents or more
% in $tok$, there is a run which puts an agent in $c_n$.
%Let $C_0 \in \cube_0$ a configuration with $2^n$ agents in $tok$.
%By induction on $i$, we cans show that $c_i$ can be reached 
%if there are $2^i$ agents or more initially in $tok$.

%%%%%%%%%%%%%%%%%%%%%%%%%%%%%
%\section{Conclusion}
%\input{conclusion}

%\noindent \textbf{Acknowledgements:} 
\paragraph*{Acknowledgements:} 
We would like to thank Javier Esparza and the anonymous reviewers for their useful feedback.

\nocite{*}
\bibliographystyle{eptcs}
\bibliography{refs}

\begin{thebibliography}{10}
\providecommand{\bibitemdeclare}[2]{}
\providecommand{\surnamestart}{}
\providecommand{\surnameend}{}
\providecommand{\urlprefix}{Available at }
\providecommand{\url}[1]{\texttt{#1}}
\providecommand{\href}[2]{\texttt{#2}}
\providecommand{\urlalt}[2]{\href{#1}{#2}}
\providecommand{\doi}[1]{doi:\urlalt{http://dx.doi.org/#1}{#1}}
\providecommand{\bibinfo}[2]{#2}

\bibitemdeclare{article}{First-Pop-Prot}
\bibitem{First-Pop-Prot}
\bibinfo{author}{Dana \surnamestart Angluin\surnameend}, \bibinfo{author}{James
  \surnamestart Aspnes\surnameend}, \bibinfo{author}{Zo{\"{e}} \surnamestart
  Diamadi\surnameend}, \bibinfo{author}{Michael~J. \surnamestart
  Fischer\surnameend} \& \bibinfo{author}{Ren{\'{e}} \surnamestart
  Peralta\surnameend} (\bibinfo{year}{2006}): \emph{\bibinfo{title}{Computation
  in networks of passively mobile finite-state sensors}}.
\newblock {\sl \bibinfo{journal}{Distributed Comput.}}
  \bibinfo{volume}{18}(\bibinfo{number}{4}), pp. \bibinfo{pages}{235--253}.
\newblock \urlprefix\url{https://doi.org/10.1007/s00446-005-0138-3}.

\bibitemdeclare{article}{Comp-Power-Pop-Prot}
\bibitem{Comp-Power-Pop-Prot}
\bibinfo{author}{Dana \surnamestart Angluin\surnameend}, \bibinfo{author}{James
  \surnamestart Aspnes\surnameend}, \bibinfo{author}{David \surnamestart
  Eisenstat\surnameend} \& \bibinfo{author}{Eric \surnamestart
  Ruppert\surnameend} (\bibinfo{year}{2007}): \emph{\bibinfo{title}{The
  computational power of population protocols}}.
\newblock {\sl \bibinfo{journal}{Distributed Comput.}}
  \bibinfo{volume}{20}(\bibinfo{number}{4}), pp. \bibinfo{pages}{279--304}.
\newblock \urlprefix\url{https://doi.org/10.1007/s00446-007-0040-2}.

\bibitemdeclare{misc}{LongVersion}
\bibitem{LongVersion}
\bibinfo{author}{A.~R. \surnamestart Balasubramanian\surnameend} \&
  \bibinfo{author}{Chana \surnamestart Weil-Kennedy\surnameend}
  (\bibinfo{year}{2021}): \emph{\bibinfo{title}{Reconfigurable Broadcast
  Networks and Asynchronous Shared-Memory Systems are Equivalent (Long
  Version)}}.
\newblock \urlprefix\url{https://arxiv.org/abs/2108.07510}.

\bibitemdeclare{inproceedings}{probtime}
\bibitem{probtime}
\bibinfo{author}{Nathalie \surnamestart Bertrand\surnameend} \&
  \bibinfo{author}{Paulin \surnamestart Fournier\surnameend}
  (\bibinfo{year}{2013}): \emph{\bibinfo{title}{Parameterized Verification of
  Many Identical Probabilistic Timed Processes}}.
\newblock In: {\sl \bibinfo{booktitle}{{IARCS} Annual Conference on Foundations
  of Software Technology and Theoretical Computer Science, {FSTTCS}}}, pp.
  \bibinfo{pages}{501--513}, \doi{10.4230/LIPIcs.FSTTCS.2013.501}.

\bibitemdeclare{inproceedings}{prob}
\bibitem{prob}
\bibinfo{author}{Nathalie \surnamestart Bertrand\surnameend},
  \bibinfo{author}{Paulin \surnamestart Fournier\surnameend} \&
  \bibinfo{author}{Arnaud \surnamestart Sangnier\surnameend}
  (\bibinfo{year}{2014}): \emph{\bibinfo{title}{Playing with Probabilities in
  Reconfigurable Broadcast Networks}}.
\newblock In: {\sl \bibinfo{booktitle}{Foundations of Software Science and
  Computation Structures - 17th International Conference, {FOSSACS}}}, pp.
  \bibinfo{pages}{134--148}, \doi{10.1007/978-3-642-54830-7\_9}.

\bibitemdeclare{inproceedings}{ICALPPatricia}
\bibitem{ICALPPatricia}
\bibinfo{author}{Patricia \surnamestart Bouyer\surnameend},
  \bibinfo{author}{Nicolas \surnamestart Markey\surnameend},
  \bibinfo{author}{Mickael \surnamestart Randour\surnameend},
  \bibinfo{author}{Arnaud \surnamestart Sangnier\surnameend} \&
  \bibinfo{author}{Daniel \surnamestart Stan\surnameend}
  (\bibinfo{year}{2016}): \emph{\bibinfo{title}{Reachability in Networks of
  Register Protocols under Stochastic Schedulers}}.
\newblock In \bibinfo{editor}{Ioannis \surnamestart
  Chatzigiannakis\surnameend}, \bibinfo{editor}{Michael \surnamestart
  Mitzenmacher\surnameend}, \bibinfo{editor}{Yuval \surnamestart
  Rabani\surnameend} \& \bibinfo{editor}{Davide \surnamestart
  Sangiorgi\surnameend}, editors: {\sl \bibinfo{booktitle}{43rd International
  Colloquium on Automata, Languages, and Programming, {ICALP} 2016, July 11-15,
  2016, Rome, Italy}}, {\sl \bibinfo{series}{LIPIcs}}~\bibinfo{volume}{55},
  \bibinfo{publisher}{Schloss Dagstuhl - Leibniz-Zentrum f{\"{u}}r Informatik},
  pp. \bibinfo{pages}{106:1--106:14}.
\newblock \urlprefix\url{https://doi.org/10.4230/LIPIcs.ICALP.2016.106}.

\bibitemdeclare{inproceedings}{ParamLiveness}
\bibitem{ParamLiveness}
\bibinfo{author}{Peter \surnamestart Chini\surnameend}, \bibinfo{author}{Roland
  \surnamestart Meyer\surnameend} \& \bibinfo{author}{Prakash \surnamestart
  Saivasan\surnameend} (\bibinfo{year}{2019}): \emph{\bibinfo{title}{Complexity
  of Liveness in Parameterized Systems}}.
\newblock In \bibinfo{editor}{Arkadev \surnamestart Chattopadhyay\surnameend}
  \& \bibinfo{editor}{Paul \surnamestart Gastin\surnameend}, editors: {\sl
  \bibinfo{booktitle}{39th {IARCS} Annual Conference on Foundations of Software
  Technology and Theoretical Computer Science, {FSTTCS} 2019, December 11-13,
  2019, Bombay, India}}, {\sl \bibinfo{series}{LIPIcs}} \bibinfo{volume}{150},
  \bibinfo{publisher}{Schloss Dagstuhl - Leibniz-Zentrum f{\"{u}}r Informatik},
  pp. \bibinfo{pages}{37:1--37:15}.
\newblock \urlprefix\url{https://doi.org/10.4230/LIPIcs.FSTTCS.2019.37}.

\bibitemdeclare{inproceedings}{Liveness}
\bibitem{Liveness}
\bibinfo{author}{Peter \surnamestart Chini\surnameend}, \bibinfo{author}{Roland
  \surnamestart Meyer\surnameend} \& \bibinfo{author}{Prakash \surnamestart
  Saivasan\surnameend} (\bibinfo{year}{2019}): \emph{\bibinfo{title}{Liveness
  in Broadcast Networks}}.
\newblock In \bibinfo{editor}{Mohamed~Faouzi \surnamestart Atig\surnameend} \&
  \bibinfo{editor}{Alexander~A. \surnamestart Schwarzmann\surnameend}, editors:
  {\sl \bibinfo{booktitle}{Networked Systems - 7th International Conference,
  {NETYS} 2019, Marrakech, Morocco, June 19-21, 2019, Revised Selected
  Papers}}, {\sl \bibinfo{series}{Lecture Notes in Computer Science}}
  \bibinfo{volume}{11704}, \bibinfo{publisher}{Springer}, pp.
  \bibinfo{pages}{52--66}.
\newblock \urlprefix\url{https://doi.org/10.1007/978-3-030-31277-0\_4}.

\bibitemdeclare{article}{FineGrained}
\bibitem{FineGrained}
\bibinfo{author}{Peter \surnamestart Chini\surnameend}, \bibinfo{author}{Roland
  \surnamestart Meyer\surnameend} \& \bibinfo{author}{Prakash \surnamestart
  Saivasan\surnameend} (\bibinfo{year}{2020}):
  \emph{\bibinfo{title}{Fine-Grained Complexity of Safety Verification}}.
\newblock {\sl \bibinfo{journal}{J. Autom. Reason.}}
  \bibinfo{volume}{64}(\bibinfo{number}{7}), pp. \bibinfo{pages}{1419--1444}.
\newblock \urlprefix\url{https://doi.org/10.1007/s10817-020-09572-x}.

\bibitemdeclare{inproceedings}{register}
\bibitem{register}
\bibinfo{author}{Giorgio \surnamestart Delzanno\surnameend},
  \bibinfo{author}{Arnaud \surnamestart Sangnier\surnameend} \&
  \bibinfo{author}{Riccardo \surnamestart Traverso\surnameend}
  (\bibinfo{year}{2013}): \emph{\bibinfo{title}{Parameterized Verification of
  Broadcast Networks of Register Automata}}.
\newblock In: {\sl \bibinfo{booktitle}{Reachability Problems - 7th
  International Workshop, {RP}}}, pp. \bibinfo{pages}{109--121},
  \doi{10.1007/978-3-642-41036-9\_11}.

\bibitemdeclare{inproceedings}{FSTTCS12}
\bibitem{FSTTCS12}
\bibinfo{author}{Giorgio \surnamestart Delzanno\surnameend},
  \bibinfo{author}{Arnaud \surnamestart Sangnier\surnameend},
  \bibinfo{author}{Riccardo \surnamestart Traverso\surnameend} \&
  \bibinfo{author}{Gianluigi \surnamestart Zavattaro\surnameend}
  (\bibinfo{year}{2012}): \emph{\bibinfo{title}{On the Complexity of
  Parameterized Reachability in Reconfigurable Broadcast Networks}}.
\newblock In \bibinfo{editor}{Deepak \surnamestart D'Souza\surnameend},
  \bibinfo{editor}{Telikepalli \surnamestart Kavitha\surnameend} \&
  \bibinfo{editor}{Jaikumar \surnamestart Radhakrishnan\surnameend}, editors:
  {\sl \bibinfo{booktitle}{{IARCS} Annual Conference on Foundations of Software
  Technology and Theoretical Computer Science, {FSTTCS} 2012, December 15-17,
  2012, Hyderabad, India}}, {\sl
  \bibinfo{series}{LIPIcs}}~\bibinfo{volume}{18}, \bibinfo{publisher}{Schloss
  Dagstuhl - Leibniz-Zentrum f{\"{u}}r Informatik}, pp.
  \bibinfo{pages}{289--300}.
\newblock \urlprefix\url{https://doi.org/10.4230/LIPIcs.FSTTCS.2012.289}.

\bibitemdeclare{inproceedings}{AdHocNetworks}
\bibitem{AdHocNetworks}
\bibinfo{author}{Giorgio \surnamestart Delzanno\surnameend},
  \bibinfo{author}{Arnaud \surnamestart Sangnier\surnameend} \&
  \bibinfo{author}{Gianluigi \surnamestart Zavattaro\surnameend}
  (\bibinfo{year}{2010}): \emph{\bibinfo{title}{Parameterized Verification of
  Ad Hoc Networks}}.
\newblock In \bibinfo{editor}{Paul \surnamestart Gastin\surnameend} \&
  \bibinfo{editor}{Fran{\c{c}}ois \surnamestart Laroussinie\surnameend},
  editors: {\sl \bibinfo{booktitle}{{CONCUR} 2010 - Concurrency Theory, 21th
  International Conference, {CONCUR} 2010, Paris, France, August 31-September
  3, 2010. Proceedings}}, {\sl \bibinfo{series}{Lecture Notes in Computer
  Science}} \bibinfo{volume}{6269}, \bibinfo{publisher}{Springer}, pp.
  \bibinfo{pages}{313--327}.
\newblock \urlprefix\url{https://doi.org/10.1007/978-3-642-15375-4\_22}.

\bibitemdeclare{article}{ModelCheckingSMS}
\bibitem{ModelCheckingSMS}
\bibinfo{author}{Antoine \surnamestart Durand{-}Gasselin\surnameend},
  \bibinfo{author}{Javier \surnamestart Esparza\surnameend},
  \bibinfo{author}{Pierre \surnamestart Ganty\surnameend} \&
  \bibinfo{author}{Rupak \surnamestart Majumdar\surnameend}
  (\bibinfo{year}{2017}): \emph{\bibinfo{title}{Model checking parameterized
  asynchronous shared-memory systems}}.
\newblock {\sl \bibinfo{journal}{Formal Methods Syst. Des.}}
  \bibinfo{volume}{50}(\bibinfo{number}{2-3}), pp. \bibinfo{pages}{140--167}.
\newblock \urlprefix\url{https://doi.org/10.1007/s10703-016-0258-3}.

\bibitemdeclare{article}{JACM16}
\bibitem{JACM16}
\bibinfo{author}{Javier \surnamestart Esparza\surnameend},
  \bibinfo{author}{Pierre \surnamestart Ganty\surnameend} \&
  \bibinfo{author}{Rupak \surnamestart Majumdar\surnameend}
  (\bibinfo{year}{2016}): \emph{\bibinfo{title}{Parameterized Verification of
  Asynchronous Shared-Memory Systems}}.
\newblock {\sl \bibinfo{journal}{J. {ACM}}}
  \bibinfo{volume}{63}(\bibinfo{number}{1}), pp. \bibinfo{pages}{10:1--10:48}.
\newblock \urlprefix\url{https://doi.org/10.1145/2842603}.

\bibitemdeclare{inproceedings}{EsparzaGMW18}
\bibitem{EsparzaGMW18}
\bibinfo{author}{Javier \surnamestart Esparza\surnameend},
  \bibinfo{author}{Pierre \surnamestart Ganty\surnameend},
  \bibinfo{author}{Rupak \surnamestart Majumdar\surnameend} \&
  \bibinfo{author}{Chana \surnamestart Weil{-}Kennedy\surnameend}
  (\bibinfo{year}{2018}): \emph{\bibinfo{title}{Verification of Immediate
  Observation Population Protocols}}.
\newblock In: {\sl \bibinfo{booktitle}{{CONCUR}}}, {\sl
  \bibinfo{series}{LIPIcs}} \bibinfo{volume}{118}, \bibinfo{publisher}{Schloss
  Dagstuhl - Leibniz-Zentrum fuer Informatik}, pp.
  \bibinfo{pages}{31:1--31:16}.
\newblock \urlprefix\url{https://doi.org/10.4230/LIPIcs.CONCUR.2018.31}.

\bibitemdeclare{inproceedings}{conf/apn/EsparzaRW19}
\bibitem{conf/apn/EsparzaRW19}
\bibinfo{author}{Javier \surnamestart Esparza\surnameend},
  \bibinfo{author}{Mikhail \surnamestart Raskin\surnameend} \&
  \bibinfo{author}{Chana \surnamestart Weil-Kennedy\surnameend}
  (\bibinfo{year}{2019}): \emph{\bibinfo{title}{Parameterized Analysis of
  Immediate Observation Petri Nets.}}
\newblock In: {\sl \bibinfo{booktitle}{Lecture Notes in Computer Science}},
  \bibinfo{volume}{11522}, pp. \bibinfo{pages}{365--385}.
\newblock \urlprefix\url{https://doi.org/10.1007/978-3-030-21571-2\_20}.

\bibitemdeclare{inproceedings}{EsparzaRW19}
\bibitem{EsparzaRW19}
\bibinfo{author}{Javier \surnamestart Esparza\surnameend},
  \bibinfo{author}{Mikhail~A. \surnamestart Raskin\surnameend} \&
  \bibinfo{author}{Chana \surnamestart Weil{-}Kennedy\surnameend}
  (\bibinfo{year}{2019}): \emph{\bibinfo{title}{Parameterized Analysis of
  Immediate Observation Petri Nets}}.
\newblock In \bibinfo{editor}{Susanna \surnamestart Donatelli\surnameend} \&
  \bibinfo{editor}{Stefan \surnamestart Haar\surnameend}, editors: {\sl
  \bibinfo{booktitle}{Application and Theory of Petri Nets and Concurrency -
  40th International Conference, {PETRI} {NETS} 2019, Aachen, Germany, June
  23-28, 2019, Proceedings}}, {\sl \bibinfo{series}{Lecture Notes in Computer
  Science}} \bibinfo{volume}{11522}, \bibinfo{publisher}{Springer}, pp.
  \bibinfo{pages}{365--385}.
\newblock \urlprefix\url{https://doi.org/10.1007/978-3-030-21571-2\_20}.

\bibitemdeclare{inproceedings}{ModelCheckingPushdown}
\bibitem{ModelCheckingPushdown}
\bibinfo{author}{Marie \surnamestart Fortin\surnameend}, \bibinfo{author}{Anca
  \surnamestart Muscholl\surnameend} \& \bibinfo{author}{Igor \surnamestart
  Walukiewicz\surnameend} (\bibinfo{year}{2017}):
  \emph{\bibinfo{title}{Model-Checking Linear-Time Properties of Parametrized
  Asynchronous Shared-Memory Pushdown Systems}}.
\newblock In \bibinfo{editor}{Rupak \surnamestart Majumdar\surnameend} \&
  \bibinfo{editor}{Viktor \surnamestart Kuncak\surnameend}, editors: {\sl
  \bibinfo{booktitle}{Computer Aided Verification - 29th International
  Conference, {CAV} 2017, Heidelberg, Germany, July 24-28, 2017, Proceedings,
  Part {II}}}, {\sl \bibinfo{series}{Lecture Notes in Computer Science}}
  \bibinfo{volume}{10427}, \bibinfo{publisher}{Springer}, pp.
  \bibinfo{pages}{155--175}.
\newblock \urlprefix\url{https://doi.org/10.1007/978-3-319-63390-9\_9}.

\bibitemdeclare{phdthesis}{Thesis}
\bibitem{Thesis}
\bibinfo{author}{Paulin \surnamestart Fournier\surnameend}
  (\bibinfo{year}{2015}): \emph{\bibinfo{title}{Parameterized verification of
  networks of many identical processes. (V{\'{e}}rification
  param{\'{e}}tr{\'{e}}e de r{\'{e}}seaux compos{\'{e}}s d'une multitude de
  processus identiques)}}.
\newblock Ph.D. thesis, \bibinfo{school}{University of Rennes 1, France}.
\newblock \urlprefix\url{https://tel.archives-ouvertes.fr/tel-01355847}.

\bibitemdeclare{inproceedings}{SafetyAlmostAlways}
\bibitem{SafetyAlmostAlways}
\bibinfo{author}{Salvatore \surnamestart {La Torre}\surnameend},
  \bibinfo{author}{Anca \surnamestart Muscholl\surnameend} \&
  \bibinfo{author}{Igor \surnamestart Walukiewicz\surnameend}
  (\bibinfo{year}{2015}): \emph{\bibinfo{title}{Safety of Parametrized
  Asynchronous Shared-Memory Systems is Almost Always Decidable}}.
\newblock In \bibinfo{editor}{Luca \surnamestart Aceto\surnameend} \&
  \bibinfo{editor}{David \surnamestart de~Frutos{-}Escrig\surnameend}, editors:
  {\sl \bibinfo{booktitle}{26th International Conference on Concurrency Theory,
  {CONCUR} 2015, Madrid, Spain, September 1.4, 2015}}, {\sl
  \bibinfo{series}{LIPIcs}}~\bibinfo{volume}{42}, \bibinfo{publisher}{Schloss
  Dagstuhl - Leibniz-Zentrum f{\"{u}}r Informatik}, pp.
  \bibinfo{pages}{72--84}.
\newblock \urlprefix\url{https://doi.org/10.4230/LIPIcs.CONCUR.2015.72}.

\end{thebibliography}

%\appendix
%\input{appendix}

\end{document}